\documentclass[reqno]{amsart}
\usepackage[top=0.5in,bottom=0.5in,left=0.5in,right=0.5in]{geometry}
\usepackage{amscd}

\usepackage{multicol}

\usepackage[T1]{fontenc}
\usepackage[latin1]{inputenc}

\usepackage[arrow,curve,matrix]{xy}
\usepackage{times}
\usepackage{graphicx}
\usepackage{extarrows}
\usepackage{flafter}
\usepackage{placeins}
\usepackage{enumitem}
\usepackage[OT2,T1]{fontenc}
\DeclareSymbolFont{cyrletters}{OT2}{wncyr}{m}{n}
\DeclareMathSymbol{\Sha}{\mathalpha}{cyrletters}{"58}

\usepackage{centernot}

\usepackage{mathrsfs}
\usepackage[color]{showkeys}
\usepackage{url}

\definecolor{refkey}{rgb}{1,1,1}
\definecolor{labelkey}{rgb}{1,1,1}
\definecolor{cite}{rgb}{0.9451,0.2706,0.4941}
\definecolor{ruri}{rgb}{0.0078,0.4022,0.8010}

\usepackage[%
bookmarks=true,bookmarksnumbered=true,%
colorlinks=true,linkcolor=ruri,citecolor=red%
 ]{hyperref}

\def\Vol{{\rm Vol}}

\def\Tr{{\rm Tr}}

\def\Tr{{\rm Tr}}

\theoremstyle{plain}

\newtheorem{theorem}{Theorem}[section]

\newtheorem{proposition/example}[theorem]{Proposition/Example}
\newtheorem{proposition}[theorem]{Proposition}
\newtheorem{corollary}[theorem]{Corollary}
\newtheorem{lemma}[theorem]{Lemma}

\theoremstyle{definition}
\newtheorem{definition}[theorem]{Definition}
\newtheorem{remark}[theorem]{Remark}
\newtheorem{example}[theorem]{Example}

\newtheorem{conjecture/question}[theorem]{Conjecture/Question}

\newtheorem{remark/definition}[theorem]{Remark/Definition}
\newtheorem{definition/notation}[theorem]{Definition/Notation}

 \theoremstyle{remark}

\numberwithin{equation}{section}

\usepackage{extarrows}\usepackage{bm}

\begin{document}
\title{\textbf{$U(1)$-Gauge  Theories on $G_2$-Manifolds}}

\author{Zhi Hu}

\address{ \textsc{School of Mathematics, Nanjing University of Science and Technology, Nanjing 230026, China}\endgraf\textsc{Department of Mathematics, Mainz University, 55128 Mainz, Germany}}

\email{halfask@mail.ustc.edu.cn; huz@uni-mainz.de}

\author{Runhong Zong}

\address{ \textsc{Department of Mathematics, Nanjing University}}

\email{rzong@nju.edu.cn}

\maketitle

\begin{abstract}
In this paper, we investigate two types of  $U(1)$-gauge field theories on $G_2$-manifolds. One  is the $U(1)$-Yang-Mills theory which admits the classical instanton solutions, we show that $G_2$-manifolds emerge from the anti-self-dual $U(1)$ instantons, which is an analogy of Yang's result for Calabi-Yau manifolds. The other one is the higher-order $U(1)$-Chern-Simons theory as a generalization of K\"{a}hler-Chern-Simons theory. We introduce the notion of higher-order $U(1)$-instanton, as the vacuum configurations of higher-order $U(1)$-Chern-Simons theory.  By suitable choice of gauge and regularization technique, we calculate the partition function under semiclassical approximation. Finally, to make sure of  the invariance  at quantum level under the large gauge transformations,  we  use Deligne-Beilinson cohomology theory to  give the the higher-order $U(1)$-Chern-Simons actions ($U(1)$-BF-type actions) for nontrivial $U(1)$-principle bundles.
\end{abstract}

\tableofcontents
\section{Introduction}
$G_2$-manifolds appear in the compactification of $M$-theory or 11-dimensional supergeravity to achieve effective 4-dimensional theory  with $\mathcal{N}=1$ supersymmetry \cite{1,a}.
Mathematically, there are two equivalent approaches to define $G_2$-manifolds. The first definition treats a $G_2$
manifold as a 7-dimensional Riemannian manifold with holonomy group as a subgroup  of $G_2$, and  the other one defines a  $G_2$-manifold as a 7-dimensional oriented spin manifold with a torsion free $G_2$-structure which is a special 3-form $\varphi$ (called fundamental 3-form) parallel with respect to the induced Levi-Civita connection.
Many examples of 7-dimensional manifold with  holonomy group  $G_2$ have also been constructed \cite{2,3,4,5,6}.
To some extent, $G_2$-manifolds can be viewed as the analog of Calabi-Yau 3-folds:
\begin{align*}
  \begin{tabular}{|c|c|}
  \hline
 Calabi-Yau 3-fold $N$  & $G_2$-manifold $M$ \\
\hline
 K\"{a}hler form $\omega$ & fundamental 3-form $\varphi$\\
\hline
 complex structure $I$ & fundamental 4-form $*_\varphi\varphi$ \\
\hline
$I$-holomorphic curve in $N$ &  associated 3-dimensional submanifold of $M$\\
\hline
special Lagrangian submanifold of $N$   & coassociated 4-dimensional submanifold of $M$ \\
  \hline
\end{tabular}.
\end{align*}
It is noteworthy that the metric on a $G_2$-manifold is totally determined by the fundamental 3-form $\varphi$ via a highly nontrivial manner, hence the fundamental 4-form $*_\varphi\varphi$ is not independent of the  fundamental 3-form $\varphi$, where $*_\varphi$ denote the Hodge star with respect to the induced metric $g_\varphi$.
Conceptually, we should  consider which results for Calabi-Yau manifolds can be generalized to $G_2$-manifolds.

In a series of papers \cite{66,7,8,88,2011,9,2012,2013,20131,99} by H.-S. Yang and his collaborators, the authors proposed a kind  of \emph{emergent gravity}, which is achieved  by considering the deformation of a symplectic manifold. In this framework,   Darboux theorem or the Moser lemma in symplectic geometry is reinterpreted as equivalence principle. From this point of view, a line bundle over a symplectic manifold leads to  a dynamical symplectic manifold described by a gauge theory of symplectic gauge fields. Then the  quantization of the dynamical symplectic manifold gives rise to a dynamical noncommutative spacetime described by a noncommutative $U(1)$-gauge theory.
 A basic  idea  of  Yang's emergent  gravity  is  to  realize  the  gauge/gravity duality using  the  Lie  algebra homomorphism between noncommutative $\star$-algebra (gauge theory side) and derivation algebra (gravity side).
Then H.-S. Yang showed that   the \emph{commutative limit} of noncommutative anti-self-dual $U(1)$-instanton equations  turns into  the equations for spin connections of an  emergent  Calabi-Yau  manifold  \cite{994}.

In Sec. 2, we will generalize such mechanism of producing Calabi-Yau  manifolds to the $G_2$-manifolds. $G_2$-instanton equations were first introduced in  \cite{sd}, which are also divided into self-dual type and anti-self-dual type according to the irreducible representations of $G_2$ on 2-forms. In the spirit of  taking commutative limit,
the anti-self-dual  $U(1)$-instanton equations give rise to  equations satisfied by a collection of local orthogonal frame fields (siebenbein), we will show that  these equations force the spin connection to be valued in the Lie algebra $\mathfrak{g}_2$ of $G_2$, hence determine a $G_2$-manifold. Namely, we prove the following theorem.
\begin{theorem}[=Theorem \ref{nko}]
A 7-dimensional  spin  Riemannian manifold $M$ is a $G_2$-manifold if and only if it is an anti-self-dual $U(1)$-instanton.
\end{theorem}

Chern-Simons theory is another important kind of gauge theory \cite{wit,witt,ad,mm}. A natural generalization  of usual Chern-Simons theory in three dimensions to a $G_2$-manifold $M$ is to consider the following action which was  first introduced in \cite{dt}
\begin{align}\label{as}
  S_{\textrm{CS}}=\int_M CS_3(A)\wedge *_\varphi\varphi,
\end{align}
where
\begin{align}
  CS_3(A)=\Tr(A dA+\frac{2}{3}A^3)
\end{align}
is the  Chern-Simons 3-form associated to the connection $A$ on a  trivial $G$-principal bundle over $M$ for a  compact  complex matrix Lie group $G$.
The critical points of the action \eqref{as} are exactly the anti-self-dual $G_2$-instantons \footnote{Some authors call those $G_2$-instantons.}, namely the solutions of the equation
\begin{align*}
  F_A\wedge*_\varphi\varphi=0,
\end{align*}
where $F_A=dA+A^2$ is the curvature 2-form of $A$. The constructions of anti-self-dual $G_2$-instantons on some special $G_2$-manifolds have been obtained \cite{ear1,ear2,ear3,ear4,ear5}.
If one choose a special $G_2$-manifold $M=CY_3\times S^1$, by dimensional reduction, $S_{\textrm{CS}}$ can be reexpressed as the sum of  $B$-model 6-brane and $\bar B$-model 6-brane actions and an extra term related to the stability of the brane \cite{bbb}. Similar to the work in \cite{bw}, one may connect the partition function of $S_{\textrm{CS}}$  with the cohomology of  moduli space of Hermitian-Yang-Mills connection on $CY_3$ (this work will appear elsewhere).

In the rest part of this  paper, we will consider \emph{higher-order Chern-Simons theory} on $\mathcal{M}=M\times L$ with $M$ being a $G_2$-manifold and $L$ standing for  $\mathbb{R}$
or $S^1$.
Our staring point is  the K\"{a}hler-Chern-Simons (KCS) theory proposed by Nair and Schiff, so we first briefly introduce this theory. The  action is given by
\begin{align}\label{ass}
  S_{\textrm{KCS}}=\int_{K\times \mathbb{R}}(CS_3(\mathcal{A})\wedge \omega+\Tr(F_\mathcal{A}\wedge \phi+F_\mathcal{A}\wedge \phi^*)),
\end{align}
where $K$ is a K\"{a}hler surface with the K\"{a}hler form $\omega$, $\mathcal{A}$ is the connection on the   trivial $G$-principal bundle over $K\times \mathbb{R}$,  $\phi$ and $\phi^*$, respectively, are  two Lagrange multipliers that are  Lie-algebra valued $(2,0)$-form and $(0,2)$-form on $K$ and also 1-forms on $\mathbb{R}$. Write $\mathcal{A}=A+B$ with $A$ being  a Lie-algebra valued 1-form on $\mathbb{R}$ and $B$ being  a Lie-algebra valued 1-form on $K$, then  assuming $K$ closed, the critical configurations of the action \eqref{ass} are solutions of equations
\begin{align*}
\frac{\partial B}{\partial t}&=0,\\
  F_B^{2,0}= F_B^{0,2}&=0,\\
  F_B\wedge\omega&=0,
\end{align*}
 where $t$ is the coordinate on $\mathbb{R}$, $F_B^{2,0}$ and $F_B^{0,2}$ stand for the $(2,0)$-part and $(0,2)$-part of the curvature 2-form $F_B$ on $K$ of $B$, respectively. Note that  these equations are equivalent to the  Hermitian-Yang-Mills equation on $K$, in particular, if $G$ is semisimple, they are also equivalent to the  anti-self-dual Yang-Mills instanton equation on $K$ (see Proposition 1.2.2 in \cite{t}). The quantization  moduli space of critical configurations has been investigated in \cite{k1,k2,li}.

Now we discuss  the $G_2$-analogue of KCS theory. Firstly, the K\"{a}hler form $\omega$ in the action \eqref{ass} is replaced by the fundamental 3-form $\varphi$, then taking into account the dimensions, $CS_3(\mathcal{A})$ should be  replaced by Chern-Simons 5-form $CS_5(\mathcal{A})$.
More precisely, we  consider the so-called  higher-order Chern-Simons action,
\begin{align}\label{aaa}
 S_{\textrm{HCS}}=\int_{\mathcal{M}}(\Tr(\mathcal{A} d\mathcal{A} d\mathcal{A}+\frac{3}{2}\mathcal{A}^3d\mathcal{A}+\frac{3}{5}\mathcal{A}^5))\wedge\varphi,
\end{align}
where $\mathcal{A}$ is the connection on the   trivial $G$-principal bundle  over $\mathcal{M}$.
Decomposing $\mathcal{A}=A+B$ for $A=\mathcal{A}_0 dt$, $B=\mathcal{A}_\mu dX^\mu$ with local  coordinates $t$ on $L$ and $\{X^\mu,i=1,\cdots,7\}$ on $M$,  the  action  \eqref{aaa} reduces to
\begin{align}\label{mv}
 S_{\textrm{HCS}}=& -\int_L dt\int_M \Tr((Bd_MB+d_MBB+\frac{3}{2}B^3)\dot{B})\wedge \varphi\nonumber\\
& + 3\int_L dt\int_M \Tr(\mathcal{A}_0 F_B^2)\wedge \varphi,
\end{align}
where $d_M$ stands for the exterior differential operator on $M$,  $\dot{B}=\frac{\partial\mathcal{A}_\mu}{\partial t} dX^\mu$, $F_B=d_MB+B^2$.
By Chern-Weil theory, the integral
$\int_M \Tr(F_B^2)\wedge \varphi$ is a topological invariant, hence the action $S_{\textrm{HCS}}$ has a large gauge symmetry at the quantum level $ \mathcal{A}_0\rightarrow  \mathcal{A}_0+f(t)$ with $\int_Lf(t)dt\in \mathbb{Z}$.
Usually, the  invariance at quantum level under large gauge transformations  requires $\varphi$ should be of integral period, i.e. $\varphi\in H^3(M,\mathbb{Z})\hookrightarrow H^3(\mathcal{M},\mathbb{Z})$ \cite{hh}. In particular, if one picks  $G=U(1)$, the  action \eqref{aaa} is more simple as
\begin{align*}
S_{\textrm{HCS}}=\int_{\mathcal{M}}\mathcal{A}\wedge  d\mathcal{A}\wedge  d\mathcal{A}\wedge\varphi,
\end{align*}
which can be generalized to the  higher-order $U(1)$-BF-type action
\begin{align}
 S_{\textrm{ABC}}=\int_{\mathcal{M}}\mathcal{A}\wedge  d\mathcal{B}\wedge  d\mathcal{C}\wedge\varphi
\end{align}
for one forms $\mathcal{A},\mathcal{B},\mathcal{C}$ on $\mathcal{M}$. It would be more appropriate to be called the  $U(1)$-ABC action.

In Sec. 3, we  introduces the notion of  \emph{higher-order $G$-instanton}. It is described by
the equation
\begin{align*}
 F_B^2\wedge \varphi=0,
\end{align*}
which is obtained by the variation of the action \eqref{mv}
 with respect to  $\mathcal{A}_0$. According to the  decomposition of $F_B^2$  into the irreducible  representations of $G_2$, we can define self-dual higher-order $G$-instantons and anti-self-dual higher-order $G$-instantons:
 \begin{align*}
   \left\{
     \begin{array}{ll}
     F_B^2=\alpha\wedge \varphi \textrm{ for some } \alpha\in\Lambda^1(M,\mathrm{Ad}_P)  , & \hbox{self-dual higher-order $G$-instanton;} \\
       F_B^2\wedge \varphi=  *_\varphi F_B^2\wedge \varphi=0, & \hbox{anti-self-dual higher-order $G$-instanton.}
     \end{array}
   \right.
 \end{align*}
The main results of this section are summarized in the following vanishing-type theorem.
\begin{theorem}[=Corollary \ref{mk}, Theorem \ref{th}, Theorem \ref{thh}]
Assume $B$ is a higher-order $U(1)$-instanton, then
\begin{enumerate}
  \item $B$ is a self-dual instanton $\Longleftrightarrow$ $B$ is an anti-self-dual instanton $\Longleftrightarrow$ $B$ is a flat connection $\Longleftrightarrow$ $F_B\wedge \varphi=0$;
  \item $B$ is a self-dual higher-order $U(1)$-instanton $\Longleftrightarrow$ $F_B^2=0$;
  \item when $B$ is a special higher-order $U(1)$-instanton, $B$ is a flat connection $\Longleftrightarrow$ $B\wedge *_\varphi F_B=0$.
\end{enumerate}
\end{theorem}

In Sec.4, we calculate the semiclassical partition function around higher-order $U(1)$-instantons for higher-order $U(1)$-Chern-Simons action over $M\times S^1$. Some crucial ingredients include:
\begin{itemize}
  \item torus gauge fixing is imposed \cite{bt,hn};
  \item heat kernel regularization  technique is used \cite{bt,li};
  \item Schwarz-Pestun-Witten's method for calculating partition function of \emph{degenerate} functional plays an important  role \cite{sc,pw,hs,hh}.
\end{itemize}

 In Sec. 5, we   discuss how to write a correct (higher-order) Chern-Simons action which is  invariant at quantum level under the large gauge transformations  when  $G$-principal bundle $\mathcal{P}$ is nontrivial. Actually,  the story becomes more subtle and complicated.  For example, a suitable framework to cope with  the nontrivial $U(1)$-principal bundle is the
\emph{Deligne-Beilinson cohomology theory} \cite{dbb,bg,gt}.  We will use the \v{C}ech  resolution of Deligne complex to give an explicit expression of the higher-order $U(1)$-Chern-Simons action (and more generally, higher-order $U(1)$-BF-type action) for a nontrivial $U(1)$-principal bundle.

\section{$G_2$-manifolds as  anti-self-dual $U(1)$-instantons}
There is a standard $G_2$-structure on $\mathbb{R}^7$ consists of  a standard  Euclidean metric $g_0$ and a 3-form $\varphi_0$ given by
\begin{align}\label{opi}
 \varphi_0=e^{123}+e^{145}+e^{246}+e^{347}-e^{167}+e^{257}-e^{356},
\end{align}
where $\{e_1,\cdots, e_7\}$ is  an orthonormal frame that provides a coordinate $\{x^1,\cdots,x^7\}$ on $\mathbb{R}^7$, and  $e^{ijk}=e^{i}\wedge e^{j}\wedge e^{k}$ for $e^{i}$ being the dual basis of $e_i$.
The $U(1)$-Yang-Mills functional on $\mathbb{R}^7$ is given by
\begin{align*}
  S=\int_{\mathbb{R}^7}d^7xF_{ij}F_{ij},
\end{align*}
where $F_{ij}=\frac{\partial A_j}{\partial x^i}-\frac{\partial A_i}{\partial x^j}$ for the dual vector field $A=A_i(x)e^{i}$ on  $\mathbb{R}^7$. Identifying $e^{i}$ with $dx^i$, the critical point reads
\begin{align*}
  d*_{g_0}F=0,
\end{align*}
where $F=dA=\frac{1}{2}F_{ij}dx^i\wedge dx^j$, $*_{g_0}$ denotes the Hodge dual with respect to $g_0$.
 Obviously, if
\begin{align}\label{tr}
  *_{g_0}F=c F\wedge \varphi_0
\end{align}
 for some nonzero constant $c$, then the above equation holds true automatically. There is a decomposition
\begin{align*}
 F=F_{(1)}+F_{(2)},
\end{align*}
where $F_{(1)},F_{(2)}$ satisfy
\begin{align*}
 *_{g_0}(F_{(1)}\wedge \varphi_0)&=-2F_{(1)},\\
 *_{g_0}(F_{(2)}\wedge \varphi_0)&=F_{(2)},
\end{align*}
respectively. Therefore, only when$c=-\frac{1}{2}$ or $c=1$  the equation \eqref{tr} admits nontrivial solutions, they are called  self-dual $U(1)$-instanton or  anti-self-dual $U(1)$-instanton, respectively, i.e.
\begin{align}
 \left\{
         \begin{array}{ll}
            *_{g_0}F=-\frac{1}{2} F\wedge \varphi_0, & \hbox{\textrm{self-dual $U(1)$-instanton};}\\
            *_{g_0}F=F\wedge \varphi_0, & \hbox{\textrm{anti-self-dual $U(1)$-instanton},}
         \end{array}
       \right.
\end{align}
or more explicitly
\begin{align}
 \left\{
         \begin{array}{ll}
          F_{ij}=-\frac{1}{4}{T_{ij}}^{kl}F_{kl}, & \hbox{\textrm{self-dual $U(1)$-instanton};}\\
            F_{ij}=\frac{1}{2}{T_{ij}}^{kl}F_{kl}, & \hbox{\textrm{anti-self-dual $U(1)$-instanton},} \\
 \end{array}
       \right.
\end{align}
where \begin{align}\label{8}
        {T_{ij}}^{kl}=\frac{1}{6}{\epsilon_{ij}}^{klpqr}(\varphi_0)_{pqr}.
      \end{align}
We will only focus on  anti-self-dual $U(1)$-instanton.

Consider a 7-dimensional spin manifold $M$ equipped with a Riemann metric $g$. Let $\{E_i=E_i^\mu\frac{\partial}{\partial X^\mu},i=1,\cdots,7\}$ be pointwisely linearly independent local orthogonal frame fields, i.e. so-called "siebenbein", over some neighborhood with local coordinate $\{X^\mu,\mu=1,\cdots,7\}$ in $M$, namely we have
\begin{align*}
  E_i^\mu E_j^\nu\delta^{ij}=g_{\mu\nu},
\end{align*}
and let $\{E^i=E^i_\mu d X^\mu,i=1,\cdots,7\}$ be the corresponding  dual 1-forms which satisfy
\begin{align*}
   E^i_\mu E^j_\nu g^{\mu\nu}=\delta^{ij}.
\end{align*}
The spin connection 1-form $\{{\omega^i}_{j}\}$ on $M$ is defined by
\begin{align}
  dE^i+{\omega^i}_{j}\wedge E^j&=0\label{lo},\\
  {\omega^i}_{j}+{\omega^j}_{i}&=0.
\end{align}
As done in \cite{994}, one does the following replacement (so-called commutative limit)
$$D_i:=\frac{\partial}{\partial x_i}+A_i\longrightarrow E_i,$$
then the anti-self-dual $U(1)$-instanton equation becomes
\begin{align}\label{g}
  [E_i,E_j]=\frac{1}{6}{T_{ij}}^{kl}[E_k,E_l],
\end{align}
or equivalently,
\begin{align}\label{o}
  {f_{ij}}^s=\frac{1}{6}{T_{ij}}^{kl}{f_{kl}}^s,
\end{align}
where ${f_{ij}}^k$ is defined by
\begin{align}\label{1l}
  [E_i,E_j]={f_{ij}}^kE_k.
\end{align}
 Writing ${\omega^i}_{j}=\omega^i_{kj}E^k$, the equation \eqref{lo} together with  \eqref{1l} implies
\begin{align*}
  {f_{ij}}^k=\omega^k_{ij}-\omega^k_{ji}=:\Omega^k_{ij},
\end{align*}
then the equation \eqref{o} reduces to
\begin{align}\label{ty}
  \Omega^s_{ij}=\frac{1}{6}{T_{ij}}^{kl}\Omega^s_{kl},
\end{align}
which is explicitly expressed in terms of components as
\begin{align}
  \Omega^s_{12}&= \Omega^s_{56}-\Omega^s_{47}\label{yu},\\
  \Omega^s_{13}&= \Omega^s_{57}+ \Omega^s_{46},\\
  \Omega^s_{14}&=-\Omega^s_{36}+\Omega^s_{27},\\
  \Omega^s_{15}&=-\Omega^s_{37}- \Omega^s_{26},\\
  \Omega^s_{16}&= \Omega^s_{25}+ \Omega^s_{34},\\
  \Omega^s_{17}&= \Omega^s_{35}-\Omega^s_{24},\\
  \Omega^s_{23}&= \Omega^s_{67}-\Omega^s_{45}\label{yui}.
\end{align}
\begin{definition}
Let $M$ be a 7-dimensional oriented spin manifold equipped with a Riemann metric $g$. If there exists an open cover $\{U_\alpha\}$ of $M$, and there is a siebenbein $\{E_i^{(\alpha)},i=1,\cdots,7\}$ (compatible with the orientation) over each $U_\alpha$ such that
\begin{align}\label{g5}
  [E_i^{(\alpha)},E_j^{(\alpha)}]=\frac{1}{6}{T_{ij}}^{kl}[E_k^{(\alpha)},E^{(\alpha)}_l],
\end{align}
we call $M$ an anti-self-dual $U(1)$-instanton.
\end{definition}

The covariant derivative on a spinor $\eta$ over an open neighborhood $U$ along the vector field $E_i$ is defined by
\begin{align*}
  \nabla_{E_i}\eta=E_i(\eta)+\frac{1}{2}\Omega^i_{jk}\Sigma_{jk}\bullet\eta,
\end{align*}
where  $\bullet$ denotes Clifford multiplication, and $\{\Sigma_{ij}=-\Sigma_{ji},i<j=1,\cdots,7\}$, satisfying the relations
\begin{align}\label{jk}
  [\Sigma_{ij},\Sigma_{kl}]=\Sigma_{il}\delta_{jk}+ \Sigma_{jk}\delta_{il}- \Sigma_{ik}\delta_{jl}-\Sigma_{jl}\delta_{ik},
\end{align}
forms a basis of Lie algebra $\mathfrak{spin}(7)$. It follows from the equations \eqref{yu}-\eqref{yui} that
\begin{align}\label{ft}
   \nabla_{E_i}\eta=&\ E_i(\eta)+\Omega^i_{56}(\Sigma_{56}+\Sigma_{12})\bullet\eta+\Omega^i_{47}(\Sigma_{47}-\Sigma_{12})\bullet\eta\nonumber\\
   &+\Omega^i_{57}(\Sigma_{57}+\Sigma_{13})\bullet\eta+\Omega^i_{46}(\Sigma_{46}+\Sigma_{13})\bullet\eta\nonumber\\
   &+\Omega^i_{36}(\Sigma_{36}-\Sigma_{14})\bullet\eta+\Omega^i_{27}(\Sigma_{27}+\Sigma_{14})\bullet\eta\nonumber\\
   &+\Omega^i_{37}(\Sigma_{37}-\Sigma_{15})\bullet\eta+\Omega^i_{26}(\Sigma_{26}-\Sigma_{15})\bullet\eta\nonumber\\
   &+\Omega^i_{25}(\Sigma_{25}+\Sigma_{16})\bullet\eta+\Omega^i_{34}(\Sigma_{34}+\Sigma_{16})\bullet\eta\nonumber\\
   &+\Omega^i_{35}(\Sigma_{35}+\Sigma_{17})\bullet\eta+\Omega^i_{24}(\Sigma_{24}-\Sigma_{17})\bullet\eta\nonumber\\
   &+\Omega^i_{67}(\Sigma_{67}+\Sigma_{23})\bullet\eta+\Omega^i_{45}(\Sigma_{45}-\Sigma_{23})\bullet\eta.
\end{align}

We introduce \begin{align*}
           V_1=\Sigma_{56}+\Sigma_{12},\ \ W_1=\Sigma_{47}-\Sigma_{12},\\
           V_2=\Sigma_{57}+\Sigma_{13},\ \ W_2=\Sigma_{46}+\Sigma_{13},\\
           V_3=\Sigma_{36}-\Sigma_{14},\ \ W_3=\Sigma_{27}+\Sigma_{14},\\
           V_4=\Sigma_{37}-\Sigma_{15},\ \ W_4=\Sigma_{26}-\Sigma_{15},\\
          V_5= \Sigma_{25}+\Sigma_{16},\ \ W_5=\Sigma_{34}+\Sigma_{16},\\
          V_6= \Sigma_{35}+\Sigma_{17},\ \ W_6=\Sigma_{24}-\Sigma_{17},\\
          V_7= \Sigma_{67}+\Sigma_{23},\ \ W_7=\Sigma_{45}-\Sigma_{23}.
          \end{align*}
By the relation \eqref{jk} one easily checks the following nontrivial  commutators \cite{hh2}
\begin{align*}
[V_1,V_2]=-V_7, \ \ \ \ \ \ &  [V_1,V_3]=V_6+W_6,\\
[V_1,V_4]=V_5,   \ \ \ \ \ \ &
  [V_1,V_5]=-2W_4\\
 [V_1,V_6]=-V_3-W_3,\ \ \ \ \ \ & [V_1,V_7]=V_2,\\
   [V_1,W_2]=W_7,\ \ \ \ \ \ &
   [V_1,W_3]=-W_6,\\ [V_1,W_4]=2V_5,\ \ \ \ \ \ &[V_1,W_5]=-W_4,\\ [V_1,W_6]=W_3,\ \ \ \ \ \  &[V_1,W_7]=-W_2, \\ [V_2,V_3]=W_5,\ \ \ \ \ \  & [V_2,V_4]=2V_6, \\ [V_2,V_5]=-V_3-W_3, \ \ \ \ \ \ &[V_2,V_6]=-2V_4,\\
   [V_2,V_7]=-V_1,\ \ \ \ \ \ &[V_2,W_1]=W_7,\\
   [V_2,W_3]=V_5-W_5,\ \ \ \ \ \ &[V_2,W_4]=V_6,\\
  [V_2,W_5]=-V_3,\ \ \ \ \ \ &[V_2,W_7]=-W_1,\\
  [V_3,V_4]=-V_7-W_7,\ \ \ \ \ \ &[V_3,V_5]=W_2,\\
  [V_3,V_6]=V_1+W_1,\ \ \ \ \ \ &[V_3,V_7]=V_4-W_4,\\
  [V_3,W_1]=W_6,\ \ \ \ \ \ &[V_3,W_2]=-2W_5,\\
  [V_3,W_4]=-W_7,\ \ \ \ \ \ &[V_3,W_5]=2W_2,\\
  [V_3,W_6]=-W_1,\ \ \ \ \ \ &[V_3,W_7]=W_4,\\
  [V_4,V_5]=V_1,\ \ \ \ \ \ &[V_4,V_6]=2V_2,\\
  [V_4,V_7]=-V_3-W_3,\ \ \ \ \ \ &[V_4,W_1]=V_5-W_5,\\
  [V_4,W_2]=-V_6,\ \ \ \ \ \ &[V_4,W_3]=-W_7,\\
  [V_4,W_5]=V_1+W_1,\ \ \ \ \ \ &[V_4,W_6]=-V_2,\\
  [V_4,W_7]=W_3,\ \ \ \ \ \ &[V_5,V_6]=-V_7,\\
  [V_5,V_7]=V_6,\ \ \ \ \ \ &[V_5,W_1]=-W_4,\\
  [V_5,W_2]=V_3,\ \ \ \ \ \ &[V_5,W_3]=-V_2+W_2,\\
  [V_5,W_4]=-2V_1,\ \ \ \ \ \ &[V_5,W_6]=V_7+W_7,\\
  [V_5,W_7]=-V_6-W_6,\ \ \ \ \ \ &[V_6,V_7]=-V_5,\\
  [V_6,W_1]=-W_3,\ \ \ \ \ \ &[V_6,W_2]=V_4,\\
  [V_6,W_3]=W_1,\ \ \ \ \ \ &[V_6,W_4]=-V_2,\\
  [V_6,W_5]=V_7+W_7,\ \ \ \ \ \ &[V_6,W_7]=V_5-W_5,\\
  [V_7,W_1]=W_2,\ \ \ \ \ \ &[V_7,W_2]=-W_1,\\
  [V_7,W_3]=-V_4+W_4,\ \ \ \ \ \ &[V_7,W_4]=-V_3-W_3,\\
  [V_7,W_5]=W_6,\ \ \ \ \ \ &[V_7,W_6]=-W_5,\\
  [W_1, W_2]=V_7,\ \ \ \ \ \ &[W_1,W_3]=2W_6,\\
  [W_1,W_4]=-V_5,\ \ \ \ \ \ &[W_1,W_5]=-V_4+W_4,\\
  [W_1,W_6]=-2W_3,\ \ \ \ \ \ &[W_1,W_7]=V_2,\\
  [W_2,W_3]=-W_5,\ \ \ \ \ \ &[W_2,W_4]=V_6+W_6,\\
  [W_2,W_5]=-2V_3,\ \ \ \ \ \ &[W_2,W_6]=V_4-W_4,\\
  [W_2,W_7]=V_1,\ \ \ \ \ \ &[W_3,W_4]=V_7+W_7,\\
  [W_3,W_5]=-W_2,\ \ \ \ \ \ &[W_3,W_6]=2W_1,\\
  [W_3,W_7]=-V_4,\ \ \ \ \ \ &[W_4,W_5]=V_1,\\
  [W_4,W_6]=-V_2+W_2,\ \ \ \ \ \ &[W_4,W_7]=-V_3,\\
  [W_5,W_6]=V_7,\ \ \ \ \ \ &[W_5,W_7]=V_6+W_6,\\
  [W_6,W_7]=V_5-W_5,\ \ \ \ \ \ &
\end{align*}
then we immediately find that  $\{V_1,\cdots, V_7,W_1,\cdots,W_7\}$ generate a 14-dimensional Lie subalgebra $\mathfrak{g}$ of $\mathfrak{spin}(7)$.

To show this subalgebra $\mathfrak{g}$ is exactly $\mathfrak{g}_2$, we only need recover to  $\varphi_0$ from the invariant spinor.
Dirac gamma matrices satisfying Clifford algebra in 7-dimensional Euclidean space  are given
by
\begin{align}
\left\{
  \begin{array}{ll}
    \Gamma_i= \Gamma_i^{(6)}, & \hbox{$i=1,\cdots,6$;} \\
    \Gamma_7=\sqrt{-1}\Gamma_1^{(6)}\cdots \Gamma_6^{(6)},
  \end{array}
\right.
\end{align}
where $\Gamma_[^{(6)}, i=1,\cdots,6,$ denote the Dirac gamma matrices in six dimensions. Choosing  purely imaginary $8\times8$ matrices, one explicitly writes Dirac gamma matrices as follows
\begin{align*}
  \Gamma_1&=\left(
             \begin{array}{cccccccc}
               0 & 0 & 0 & 0 & 0 & 0 & 0& \sqrt{-1} \\
               0 & 0 & \sqrt{-1} & 0 & 0 & 0 & 0& 0 \\
               0 & -\sqrt{-1} & 0 & 0 & 0 & 0 & 0& 0 \\
                0 & 0 & 0 & 0 & 0 & 0 & \sqrt{-1}& 0 \\
                0 & 0 & 0 & 0 & 0 & -\sqrt{-1} & 0& 0 \\
                0 & 0 & 0 & 0 & \sqrt{-1} & 0 & 0& 0 \\
                0 & 0 & 0 & -\sqrt{-1} & 0 & 0 & 0& 0 \\
               -\sqrt{-1}& 0 & 0 & 0 & 0 & 0 & 0& 0 \\
             \end{array}
           \right),\end{align*}
           \begin{align*}
           \Gamma_2&=\left(
             \begin{array}{cccccccc}
               0 & 0 & -\sqrt{-1} & 0 & 0 & 0 & 0& 0 \\
               0 & 0 & 0& 0 & 0 & 0 & 0& \sqrt{-1} \\
              \sqrt{-1} & 0 & 0 & 0 & 0 & 0 & 0& 0 \\
                0 & 0 & 0 & 0 & 0 & \sqrt{-1} & 0& 0 \\
                0 & 0 & 0 & 0 & 0 & 0 & \sqrt{-1}& 0 \\
                0 & 0 & 0 & -\sqrt{-1} &0& 0 & 0& 0 \\
                0 & 0 & 0 & 0 & -\sqrt{-1} & 0 & 0& 0 \\
              0& -\sqrt{-1} & 0 & 0 & 0 & 0 & 0& 0 \\
             \end{array}
           \right),\end{align*}
         \begin{align*} \Gamma_3&=\left(
             \begin{array}{cccccccc}
               0 & \sqrt{-1}& 0 & 0 & 0 & 0 & 0& 0 \\
           -\sqrt{-1} & 0 & 0 & 0 & 0 & 0 & 0& 0 \\
               0 & 0 & 0 & 0 & 0 & 0 & 0& \sqrt{-1} \\
                0 & 0 & 0 & 0 & -\sqrt{-1} & 0 & 0& 0 \\
                0 & 0 & 0 & \sqrt{-1} & 0 & 0& 0& 0 \\
                0 & 0 & 0 & 0 & 0 & 0 & \sqrt{-1}& 0 \\
                0 & 0 & 0 &0& 0 & -\sqrt{-1} & 0& 0 \\
              0& 0 & -\sqrt{-1} & 0 & 0 & 0 & 0& 0 \\
             \end{array}
           \right),\end{align*}
          \begin{align*} \Gamma_4&=\left(
             \begin{array}{cccccccc}
               0 & 0 & 0 & 0 & 0 & 0 & -\sqrt{-1}& 0 \\
               0 & 0 & 0& 0 & 0 & -\sqrt{-1}& 0& 0 \\
             0 & 0 & 0 & 0 & \sqrt{-1} & 0 & 0& 0 \\
                0 & 0 & 0 & 0 & 0 & 0 & 0& \sqrt{-1} \\
                0 & 0 & -\sqrt{-1} & 0 & 0 & 0 & 0& 0 \\
                0 & \sqrt{-1} & 0 & 0 &0& 0 & 0& 0 \\
                \sqrt{-1} & 0 & 0 & 0 & 0 & 0 & 0& 0 \\
              0& 0 & 0 & -\sqrt{-1} & 0 & 0 & 0& 0 \\
             \end{array}
           \right),\end{align*}
          \begin{align*} \Gamma_5&=\left(
             \begin{array}{cccccccc}
               0 & 0 & 0 & 0 & 0 & \sqrt{-1} & 0& 0 \\
              0 & 0 & 0 & 0 & 0 & 0 & -\sqrt{-1}& 0 \\
               0 & 0 & 0 & -\sqrt{-1} & 0 & 0 & 0& 0 \\
                0 & 0 & \sqrt{-1} & 0 & 0 & 0 & 0& 0 \\
                0 & 0 & 0 & 0 & 0 & 0& 0& \sqrt{-1} \\
                -\sqrt{-1}& 0 & 0 & 0 & 0 & 0 & 0& 0 \\
                0 & \sqrt{-1} & 0 &0& 0 & 0 & 0& 0 \\
              0& 0 & 0 & 0 & -\sqrt{-1} & 0 & 0& 0 \\
             \end{array}
           \right),\end{align*}
          \begin{align*} \Gamma_6&=\left(
             \begin{array}{cccccccc}
               0 & 0 & 0 & 0 & -\sqrt{-1} & 0 & 0& 0 \\
               0 & 0 & 0& \sqrt{-1} & 0 & 0 & 0& 0 \\
             0 & 0 & 0 & 0 & 0 & 0 & -\sqrt{-1}& 0 \\
                0 & -\sqrt{-1} & 0 & 0 & 0 & 0 & 0& 0 \\
               \sqrt{-1} & 0 & 0 & 0 & 0 & 0 & 0& 0 \\
                0 & 0 & 0 & 0 &0& 0 & 0& \sqrt{-1}\\
                0 & 0 & \sqrt{-1} & 0 & 0 & 0 & 0& 0 \\
              0& 0 & 0 & 0 & 0 & -\sqrt{-1} & 0& 0 \\
             \end{array}
           \right),\end{align*}
           \begin{align*}  \Gamma_7&=\left(
             \begin{array}{cccccccc}
               0 & 0 & 0 &\sqrt{-1} & 0 & 0 & 0& 0 \\
              0 & 0 & 0 & 0 & \sqrt{-1} & 0 & 0& 0 \\
               0 & 0 & 0 & 0 & 0 & \sqrt{-1} & 0& 0 \\
                -\sqrt{-1} & 0 & 0 & 0 & 0 & 0 & 0& 0 \\
                0 & -\sqrt{-1} & 0 & 0 & 0 & 0& 0& 0 \\
                0 & 0 & -\sqrt{-1} & 0 & 0 & 0 & 0& 0 \\
                0 & 0 & 0 &0& 0 & 0 & 0& \sqrt{-1} \\
              0& 0 & 0 & 0 & 0 & 0 & -\sqrt{-1}& 0 \\
             \end{array}
           \right).\end{align*}
Then $\Sigma_{ij}$ can be realized as
\begin{align*}
 \Sigma_{ij}=\frac{1}{4}[\Gamma_i,\Gamma_j],
\end{align*}
which provides  explicit expressions
\begin{align*}
V_1=\frac{1}{2}\left(
                 \begin{array}{cccccccc}
                   0 & 1 & 0 & 0 & 0 & 0 & 0 & -1 \\
                 -1 & 0 & 1 & 0 & 0 & 0 & 0 & 0 \\
                   0 & -1 & 0 & 0 & 0 & 0 & 0 & 1 \\
                   0& 0 & 0 & 0 & 1 & 0 & 1&0 \\
                  0 & 0 & 0 & -1 & 0 & 1 & 0 & 0 \\
                   0 & 0 & 0 & 0 & -1 & 0 & -1 & 0 \\
                   0 & 0 & 0 & -1 & 0 & 1 & 0 & 0 \\
                   1 & 0 & -1 & 0 & 0 & 0 & 0 & 0 \\
                 \end{array}
               \right),\
W_1=\frac{1}{2}\left(
                 \begin{array}{cccccccc}
                   0 & -1 & 0 & 0 & 0 & 0 & 0 & 1 \\
                 1 & 0 & -1 & 0 & 0 & 0 & 0 & 0 \\
                   0 & 1 & 0 & 0 & 0 & 0 & 0 & -1 \\
                   0& 0 & 0 & 0 & -1 & 0 & 1&0 \\
                  0 & 0 & 0 & 1 & 0 & 1 & 0 & 0 \\
                   0 & 0 & 0 & 0 & -1 & 0 & 1 & 0 \\
                   0 & 0 & 0 & -1 & 0 & -1 & 0 & 0 \\
                   -1 & 0 & 1 & 0 & 0 & 0 & 0 & 0 \\
                 \end{array}
               \right),
\end{align*}
\begin{align*}
V_2=\left(
                 \begin{array}{cccccccc}
                   0 & 0 & 1 & 0 & 0 & 0 & 0 & 0 \\
                 0 & 0 & 0 & 0 & 0 & 0 & 0 & 0 \\
                   -1 & 0 & 0 & 0 & 0 & 0 & 0 & 0 \\
                   0& 0 & 0 & 0 & 0 & 0 & 0&0 \\
                  0 & 0 & 0 & 0 & 0 & 0 & 1 & 0 \\
                   0 & 0 & 0 & 0 & 0 & 0 & 0 & 0 \\
                   0 & 0 & 0 & 0 & -1 & 0 & 0 & 0 \\
                   0 & 0 & 0 & 0 & 0 & 0 & 0 & 0 \\
                 \end{array}
               \right),\
W_2=\left(
                 \begin{array}{cccccccc}
                   0 & 0 & 1 & 0 & 0 & 0 & 0 & 0 \\
                 0& 0 & 0 & 0 & 0 & 0 & 0 & 0 \\
                   -1 & 0& 0 & 0 & 0 & 0 & 0 & 0 \\
                   0& 0 & 0 & 0 & 0 & 1 & 0&0 \\
                  0 & 0 & 0 & 0 & 0 & 0 & 0 & 0 \\
                   0 & 0 & 0 & -1 & 0& 0 & 0 & 0 \\
                   0 & 0 & 0 & 0 & 0 & 0 & 0 & 0 \\
                  0 & 0 & 0 & 0 & 0 & 0 & 0 & 0 \\
                 \end{array}
               \right),
\end{align*}

\begin{align*}
V_3=\left(
                 \begin{array}{cccccccc}
                   0 & 0 & 0 & -1 & 0 & 0 & 0 & 0 \\
                 0& 0 & 0 & 0 & 0 & 0 & 0 & 0 \\
                   0 & 0 & 0 & 0 & 0 & 1 & 0 & 0 \\
                  1& 0 & 0 & 0 & 0 & 0 & 0&0 \\
                  0 & 0 & 0 & 0 & 0 & 0 & 0 & 0 \\
                   0 & 0 & -1 & 0 & 0 & 0 &0 & 0 \\
                   0 & 0 & 0 & 0 & 0 &0 & 0 & 0 \\
                   0 & 0 & 0 & 0 & 0 & 0 & 0 & 0 \\
                 \end{array}
               \right),\
W_3=\frac{1}{2}\left(
                 \begin{array}{cccccccc}
                   0 & 0 & 0 & 1 & 0 & 1 & 0 & 0 \\
                 0 & 0 & 0 & 0 & -1 & 0 & 1 & 0 \\
                   0 & 0 & 0 & -1 & 0 & -1 & 0 & 0 \\
                   -1& 0 & 1 & 0 & 0 & 0 & 0&0 \\
                  0 & 1 & 0 & 0 & 0 & 0 & 0 & -1 \\
                   -1 & 0 & 1 & 0 & 0 & 0 & 0 & 0 \\
                   0 & -1 & 0 & 0 & 0 & 0 & 0 & 1 \\
                  0 & 0 & 0 & 0 & 1 & 0 & -1 & 0 \\
                 \end{array}
               \right),
\end{align*}

\begin{align*}
V_4=\left(
                 \begin{array}{cccccccc}
                   0 & 0 & 0 & 0 & -1 & 0 & 0 & 0 \\
                 0& 0 & 0 & 0 & 0 & 0 & 0 & 0 \\
                   0 & 0 & 0 & 0 & 0 & 0 & 1 & 0 \\
                 0& 0 & 0 & 0 & 0 & 0 & 0&0 \\
                  1 & 0 & 0 & 0 & 0 & 0 & 0 & 0 \\
                   0 & 0 & 0 & 0 & 0 & 0 &0 & 0 \\
                   0 & 0 & -1 & 0 & 0 &0 & 0 & 0 \\
                   0 & 0 & 0 & 0 & 0 & 0 & 0 & 0 \\
                 \end{array}
               \right),\
W_4=\frac{1}{2}\left(
                 \begin{array}{cccccccc}
                   0 & 0 & 0 & 0 & -1 & 0 & -1 & 0 \\
                 0 & 0 & 0 & -1 & 0 & 1 & 0 & 0 \\
                   0 & 0 & 0 & 0 & 1 & 0 & 1 & 0 \\
                   0& 1 & 0 & 0 & 0 & 0 & 0&-1 \\
                  1 &0 & -1 & 0 & 0 & 0 & 0 & 0 \\
                   0 & -1 & 0 & 0 & 0 & 0 & 0 & 1 \\
                   1 & 0 & -1 & 0 & 0 & 0 & 0 & 0\\
                  0 & 0 & 0 & 1 & 0 & -1 & 0 & 0 \\
                 \end{array}
               \right),
\end{align*}

\begin{align*}
V_5=\frac{1}{2}\left(
                 \begin{array}{cccccccc}
                   0 & 0 & 0 & -1 & 0 & 1 & 0 & 0 \\
                 0& 0 & 0 & 0 & 1 & 0 & 1 & 0 \\
                   0 & 0 & 0 & 1 & 0 & -1 & 0 & 0 \\
                  1& 0 & -1 & 0 & 0 & 0 & 0&0 \\
                  0 & -1 & 0 & 0 & 0 & 0 & 0 & 1 \\
                   -1 & 0 & 1 & 0 & 0 & 0 &0 & 0 \\
                   0 & -1 & 0 & 0 & 0 &0 & 0 & 1 \\
                   0 & 0 & 0 & 0 & -1 & 0 & -1 & 0 \\
                 \end{array}
               \right),\
W_5=\left(
                 \begin{array}{cccccccc}
                   0 & 0 & 0 & 0 & 0 & 1 & 0 & 0 \\
                 0 & 0 & 0 & 0 & 0& 0 & 0 & 0 \\
                   0 & 0 & 0 & 1 & 0 & 0 & 0 & 0 \\
                  0& 0 & -1 & 0& 0 & 0 & 0&0 \\
                  0 & 0 & 0 & 0 & 0 & 0 & 0 & 0 \\
                   -1 & 0 & 0 & 0 & 0 & 0 & 0 & 0 \\
                   0 & 0 & 0 & 0 & 0 & 0 & 0 & 0 \\
                  0 & 0 & 0 & 0 & 0 & 0 &0 & 0 \\
                 \end{array}
               \right),
\end{align*}

\begin{align*}
V_6=\left(
                 \begin{array}{cccccccc}
                   0 & 0 & 0 & 0 & 0 & 0 & 1 & 0 \\
                 0& 0 & 0 & 0 & 0 & 0 & 0 & 0 \\
                   0 & 0 & 0 & 0 & 1 & 0 & 0 & 0 \\
                  0& 0 & 0 & 0 & 0 & 0 & 0&0 \\
                  0 & 0 & -1 & 0 & 0 & 0 & 0 & 0 \\
                   0 & 0 & 0 & 0 & 0 & 0 &0 & 0 \\
                  -1 & 0 & 0 & 0 & 0 &0 & 0 & 0 \\
                   0 & 0 & 0 & 0 & 0 & 0 & 0 & 0 \\
                 \end{array}
               \right),\
W_6=\frac{1}{2}\left(
                 \begin{array}{cccccccc}
                   0 & 0 & 0 & 0 & 1 & 0 & -1 & 0 \\
                 0 & 0 & 0 & 1 & 0 & 1 & 0 & 0 \\
                   0 & 0 & 0 & 0 & -1 & 0 & 1 & 0 \\
                   0& -1 & 0 & 0 & 0 & 0 & 0&1 \\
                  -1& 0 & 1 & 0 & 0 & 0 & 0 & 0 \\
                   0 & -1 & 0 & 0 & 0 & 0 & 0 & 1 \\
                  1 & 0 & -1& 0 & 0 & 0 & 0 & 0 \\
                  0 & 0 & 0 & -1 & 0 & -1 & 0 & 0 \\
                 \end{array}
               \right),
\end{align*}

\begin{align*}
V_7=\frac{1}{2}\left(
                 \begin{array}{cccccccc}
                   0 & -1 & 0 & 0 & 0 & 0 & 0 & 1\\
                 1& 0 & 1 & 0 & 0 & 0 & 0 & 0 \\
                   0 & -1 & 0 & 0 & 0 & 0 & 0 & 1 \\
                  0& 0 & 0 & 0 & 1 & 0 & -1&0 \\
                  0 & 0 & 0 & -1 & 0 & 1 & 0 & 0 \\
                   0 & 0 & 0 & 0 & -1 & 0 &1& 0 \\
                   0 & 0 & 0 & 1 & 0 &-1 & 0 & 0 \\
                   -1 & 0 & -1 & 0 & 0 & 0 & 0 & 0 \\
                 \end{array}
               \right),\
W_7=\frac{1}{2}\left(
                 \begin{array}{cccccccc}
                   0 & 1 & 0 & 0 & 0 & 0 & 0 & -1 \\
                 -1 & 0 & -1 & 0 & 0 & 0 & 0 & 0 \\
                   0 & 1 & 0 & 0 & 0 & 0 & 0 & -1 \\
                   0& 0 & 0& 0 & 1 & 0 & 1&0 \\
                  0 & 0 & 0 & -1 & 0 & -1 & 0 & 0 \\
                   0 & 0 & 0 & 0 & 1 & 0 & 1 & 0 \\
                   0 & 0 & 0 & -1 & 0 & -1 & 0 & 0 \\
                  1 & 0 & 1 & 0 & 0 & 0 & 0 & 0 \\
                 \end{array}
               \right).
\end{align*}

Let $G$ be a connected subgroup of $Spin(7)$ with Lie algebra $\mathfrak{g}$, then we have  the unique $G$-invariant spinor $\eta_0$ up to a constant scalar determined by
\begin{align}
  V_i\bullet\eta_0=W_i\bullet\eta_0=0
\end{align}
 for any $i=1,\cdots,7$. Imposing the normalization condition, we write
\begin{align}\label{788}
 \eta_0=\frac{1}{\sqrt{2}}\left(
        \begin{array}{c}
          0 \\
          1 \\
          0 \\
          0 \\
          0 \\
          0 \\
          0 \\
          -1 \\
        \end{array}
      \right).
\end{align}
Define \begin{align}\label{ui}
       \psi=\psi_{ijk}e^{ijk}=\sqrt{-1}\eta_0^\dagger\Gamma_{ijk}\eta_0 e^{ijk},
       \end{align}
where \begin{align*}
        \Gamma_{ijk}=\frac{1}{3!}\sum_{\sigma}(-1)^{|\sigma|}\Gamma_{\sigma(i)}\Gamma_{\sigma(j)}\Gamma_{\sigma(k)}
      \end{align*}
 with $\sigma$ standing for a permutation. It is clear that $\psi$ is $G$-invariant. A direct calculation shows that the nonzero coefficients $\psi_{ijk}=\sqrt{-1}\eta_0^\dagger\Gamma_{ijk}$ $(i<j<k)$ are given by
\begin{align*}
  &\psi_{123}=\psi_{246}=\psi_{167}=\psi_{257}=\psi_{356}=1,\\
&\psi_{347}=\psi_{145}=-1.
\end{align*}
If one reassigns the frame $\{e_1,e_2,e_3,e_6,e_7,e_4,e_5\}$ to $\{e_1,e_2,e_3,e_4,e_5,e_6,e_7\}$, we find that
$\psi$ exactly coincides with $6\varphi_0$.

If $M$ is an anti-self-dual $U(1)$-instanton, from the above arguments it follows that the condition \eqref{g5} guarantees the  spin connection 1-form is valued in $\mathfrak{g_2}$, which implies that  the holonomy group of the metric $g$ lies in $G_2$ due to Ambrose-Singer theorem.
 Conversely, if $M$ is  a $G_2$-manifold, it can be made into an anti-self-dual $U(1)$-instanton. As a consequence, we have the following theorem.
\begin{theorem}\label{nko}$M$ is a $G_2$-manifold if and only if it is an anti-self-dual $U(1)$-instanton.
\end{theorem}

\section{Higher-order $U(1)$-instantons}
Varying $\mathcal{A}_0$ in the action \eqref{mv} leads to the equation
\begin{align}\label{jh}
  F_B^2\wedge\varphi=0.
\end{align}
Here $F_B\in \Lambda^2(M,\mathrm{Ad}_P)$ for a $G$-principal bundle $P$ (not necessarily trivial) over $M$ with the adjoint vector  bundle $\mathrm{Ad}_P $.
It is known that
 the spaces $\Lambda^2(M), \Lambda^3(M)$ and $ \Lambda^4(M)$ of 2-forms, 3-forms and 4-forms on $M$ have orthogonal  decompositions with respect to the metric $g_\varphi$, respectively,
 \begin{align*}
\Lambda^2(M)&=\Lambda_{(1)}^2(M)\oplus\Lambda_{(2)}^2(M),\\
 \Lambda^3(M)&=\Lambda_{(1)}^3(M)\oplus\Lambda_{(2)}^3(M)\oplus\Lambda_{(3)}^3(M),\\
 \Lambda^4(M)&=\Lambda_{(1)}^4(M)\oplus\Lambda_{(2)}^4(M)\oplus\Lambda_{(3)}^4(M),
\end{align*}
where
\begin{align*}
\Lambda_{(1)}^2(M)&=\{\beta\in\Lambda^2:*_\varphi(\varphi\wedge \beta)=-2\beta\},\\
\Lambda_{(2)}^2(M)&=\{\beta\in\Lambda^2:*_\varphi(\varphi\wedge \beta)=\beta\},\\
\Lambda_{(1)}^3(M)&=\{f\varphi:f\in C^\infty(M)\},\\
 \Lambda_{(2)}^3(M)&=\{X\lrcorner *_\varphi\varphi:X\in TM\},\\
\Lambda_{(3)}^3(M)&=\{\beta\in \Lambda^3(M):\varphi\wedge \beta=*_\varphi\varphi\wedge\beta=0\},\\
 \Lambda_{(1)}^4(M)&=\{f*_\varphi\varphi:f\in C^\infty(M)\},\\
 \Lambda_{(2)}^4(M)&=\{\alpha\wedge \varphi:\alpha\in\Lambda^1(M)\},\\
\Lambda_{(3)}^4(M)&=\{\eta\in \Lambda^4(M):\varphi\wedge \beta=\varphi\wedge*_\varphi\beta=0\}.
\end{align*}
Then we introduce  the following definition.
\begin{definition}  The  connection $B$ on $P$
 \begin{itemize}
   \item is called a  higher-order $G$-instanton  if $F_B^2\in\Lambda_{(2)}^4(M,\mathrm{Ad}_P)\oplus \Lambda_{(3)}^4(M,\mathrm{Ad}_P)$,
   \item is called a  higher-order flat $G$-instanton  if $F_B^2=0$,
    \item is called a self-dual higher-order $G$-instanton if $F_B^2\in\Lambda_{(2)}^4(M,\mathrm{Ad}_P)$,
    \item is called an anti-self-dual higher-order $G$-instanton if $F_B^2\in\Lambda_{(3)}^4(M,\mathrm{Ad}_P)$.
 \end{itemize}
 \end{definition}
 \begin{theorem}The following two conditions are equivalent
\begin{enumerate}
  \item $B$ is a higher-order $G$-instanton,
  \item $ 2|F_B|_{g_\varphi}=|F_B\wedge \varphi|_{g_\varphi}$,
\end{enumerate}
 where   $|\bullet|_{g_\varphi}=\bullet\wedge*_\varphi\bullet$.
 \end{theorem}
 \begin{proof}(1)$\Rightarrow$ (2): Let $(F_B)_{(i)}$ be the projection of $F_B$ on $ \Lambda_{(i)}^2(M,\mathrm{Ad}_P)$.
For a higher-order $G$-instanton $B$, writing $F_B=(F_B)_{(1)}+(F_B)_{(2)}$, we obtain \begin{align*}
               F_B^2\wedge\varphi&=((F_B)_{(1)}+(F_B)_{(2)})\wedge ((F_B)_{(1)}+(F_B)_{(2)})\wedge \varphi\\
               &=-2(F_B)_{(1)}\wedge *_\varphi(F_B)_{(1)}+(F_B)_{(2)}\wedge *_\varphi(F_B)_{(2)}\\
               &=0,
             \end{align*}
namely,
\begin{align}\label{sd}
 2|(F_B)_{(1)}|_{g_\varphi}=|(F_B)_{(2)}|_{g_\varphi},
\end{align}
 which is equivalent to the identity in (2).
Indeed, one has
\begin{align*}
  |F_B\wedge \varphi|_{g_\varphi}&=|((F_B)_{(1)}+(F_B)_{(2)})\wedge \varphi|_{g_\varphi}\\
&=|-2(F_B)_{(1)}+(F_B)_{(2)}|_{g_\varphi}\\
&=
4|(F_B)_{(1)}|_{g_\varphi}+|(F_B)_{(2)}|_{g_\varphi},
\end{align*}
which immediately yields
\begin{align*}
 |F_B\wedge \varphi|_{g_\varphi}=2(|(F_B)_{(1)}|_{g_\varphi}+|(F_B)_{(2)}|_{g_\varphi})=2 |F_B|_{g_\varphi}.
\end{align*}

      (2) $\Rightarrow$ (1): It is known that
       \begin{align*}
         (F_B)_{(1)}=\frac{1}{3}F_B-\frac{1}{3}*_\varphi(\varphi\wedge F_B),\\
           (F_B)_{(2)}=\frac{2}{3}F_B+\frac{1}{3}*_\varphi(\varphi\wedge F_B).
       \end{align*}
       By means of \eqref{sd}, we have
       \begin{align*}
        2|F_B|_{g_\varphi}+2|F_B\wedge \varphi|_{g_\varphi}-4F_B^2\wedge \varphi=4|F_B|_{g_\varphi}+|F_B\wedge \varphi|_{g_\varphi}+4F_B^2\wedge\varphi,
       \end{align*}
thus \begin{align*}
        3|F_B\wedge \varphi|_{g_\varphi}-4F_B^2\wedge \varphi=3|F_B\wedge \varphi|_{g_\varphi}+4F_B^2\wedge\varphi.
       \end{align*}
       Therefore, $F_B^2\wedge \varphi=0$, i.e. $B$ is a higher-order $G$-instanton.
\end{proof}
\begin{corollary}\label{mk}Assume $G$ is a compact semisimple complex Lie group or $G=U(1)$. If $B$ is a higher-order $G$-instanton, then the following conditions are equivalent
\begin{enumerate}
  \item $B$ is a self-dual $G$-instanton,
  \item $B$ is an ant-self-dual $G$-instanton,
  \item $B$ is a flat connection,
\item $F_B\wedge \varphi=0$.
\end{enumerate}
\end{corollary}

Consider the  formal  space $\mathbb{M}_P$  consisting of higher-order $G$-instantons on $P$.
If  $a\in \Lambda^1(M,\mathrm{Ad}_P)$   satisfies the equation
\begin{align}\label{mnj}
 D_Ba\wedge F_B\wedge\varphi+ F_B\wedge D_Ba\wedge\varphi=0,
\end{align}
where $D_B=d_M+[B\wedge \bullet]$, then $a$ can be treated as a tangent  vector at $B\in\mathbb{M}_P$.
One defines a formal 2-form $\Omega$ on $ \mathbb{M}_P$ by
\begin{align}
  \Omega|_B(a_1,a_2)=\int_M \Tr(F_B\wedge (a_1\wedge a_2-a_2\wedge a_1))\wedge \varphi
\end{align}
for $a_1,a_2\in T_B\mathbb{M}_P$, and defines a formal vector field $V$ on $ \mathbb{M}_P$ as
 \begin{align*}
 V|_B=D_B\theta
 \end{align*}
for a Lie-algebra valued function $\theta$  on $M$.
\begin{proposition}
  Assume $M$ is closed.
\begin{enumerate}
    \item $\Omega$ can be viewed as a formal pre-symplectic form on  $ \mathbb{M}_P$.
    \item $V\lrcorner\Omega=0$.
  \end{enumerate}
   \end{proposition}
  \begin{proof} (1) Let $\mathfrak{d}$ denote the  exterior differential operator on  $ \mathbb{M}_P$. Then for any $a_i\in T_B\mathbb{M}_P$, $i=1,2,3$,  we calculate at $B\in\mathbb{M}_P$
  \begin{align*}
  &(\mathfrak{d}\Omega)|_B(a_1,a_2,a_3)\\
  =&\int_M \Tr(d_Ma_1\wedge a_2\wedge a_3)\wedge \varphi
  -\int_M \Tr(d_Ma_2\wedge a_1\wedge a_3)\wedge \varphi
  +\int_M \Tr(d_Ma_3\wedge a_1\wedge a_2)\wedge \varphi\\
  &-\int_M \Tr(d_Ma_1\wedge a_3\wedge a_2)\wedge \varphi
  +\int_M \Tr(d_Ma_2\wedge a_3\wedge a_1)\wedge \varphi
  -\int_M\Tr(d_Ma_3\wedge a_2\wedge a_1)\wedge \varphi\\
  &+\int_M \Tr(B\wedge a_1\wedge a_2\wedge a_3)\wedge \varphi+\int_M \Tr(a_1\wedge B\wedge a_2\wedge a_3)\wedge \varphi-\int_M \Tr(B\wedge a_1\wedge a_3\wedge a_2)\wedge \varphi\\
  &-\int_M \Tr(a_1\wedge B\wedge a_3\wedge a_2)\wedge \varphi-\int_M \Tr(B\wedge a_2\wedge a_1\wedge a_3)\wedge \varphi-\int_M \Tr(a_2\wedge B \wedge a_1\wedge a_3)\wedge \varphi\\
  &+\int_M \Tr(B\wedge a_2\wedge a_3\wedge a_1)\wedge \varphi+\int_M \Tr(a_2\wedge B \wedge a_3\wedge a_1)\wedge \varphi+\int_M \Tr( B \wedge a_3\wedge a_1\wedge a_2)\wedge \varphi\\
  &+\int_M \Tr( a_3 \wedge B\wedge a_1\wedge a_2)\wedge \varphi-\int_M \Tr(B\wedge a_3\wedge a_2\wedge a_1)\wedge \varphi-\int_M \Tr(a_3\wedge B \wedge a_2\wedge a_1)\wedge \varphi\\
  =&\ \int_M \Tr(d_M((a_1\wedge a_2-a_2\wedge a_1)\wedge a_3))\wedge \varphi\\
  =&\ 0,
\end{align*}
which means that
  $\mathfrak{d}\Omega=0$, i.e. $\Omega$ is  a formal pre-symplectic form on  $ \mathbb{M}_P$.

(2)  For $a\in T_B\mathbb{M}$, we have
\begin{align*}
 (V\lrcorner\Omega)|_B(a)&=\Omega|_B(D_B\theta,a)\\
&=\int_M\Tr(F_B\wedge (D_B\theta\wedge a-a\wedge D_B\theta))\wedge \varphi\\
&=\int_M\Tr(d_M(F_B\wedge(\theta a+a\theta) ))\wedge \varphi-\int_M\Tr(F_B\wedge (\theta D_B a+D_Ba\theta))\wedge \varphi\\
&=0,
\end{align*}
where the last equality is as the result of $a$ satisfying \eqref{mnj}. Therefore, $V\lrcorner\Omega=0$.
  \end{proof}

From now on, we always assume $G=U(1)$.

 \begin{theorem}\label{th}
  If the connection $B$ is   a self-dual higher-order $U(1)$-instanton,  then $B$ is   a   higher-order flat $U(1)$-instanton.

 \end{theorem}
\begin{proof}  We need to show  $F_B^2=0$. Assume $0\neq F_B^2\in\Lambda_{(2)}^4(M)$,  then by definition  there exists nonzero $\alpha\in\Lambda^1(M)$
                          such that \begin{align*}
                                     F_B^2=\alpha\wedge \varphi.
                                    \end{align*}
                          Pick   a point $p\in M$ such that $F_B|_p\neq 0$, and let $U$ be an open neighbourhood of $p$ with a local coordinate system such that $\varphi|_p=\varphi_0$. Since $G_2$ acts  transitively  on $S^6$, by a suitable $G_2$-action, we can put $\alpha|_p=ce^1$ for some constant $c\neq 0$. Without loss of generality, we assume $c=1$.
                           Then we have
                                         \begin{align}\label{fff}
                                           (F_B|_p)^2=-e^{1246}-e^{1347}+e^{1257}-e^{1356}.
                                         \end{align}
                                  We show that the above equality cannot occur by the following steps.

                          Step 1: Write $F_B|_p=e^1\wedge a+b$ with $a=\sum_{i\geq 2}a_ie^i$ and $b=\sum_{i<j, i\geq 2}b_{ij}e^{ij}$, then
                          \begin{align}\label{fg}
                            e^1\lrcorner (F_B|_p)^2=a\wedge b=-e^{246}-e^{347}+e^{257}-e^{356}.
                          \end{align}

                         Step 2: We claim all $a_i$'s do not vanish. Assume $a_2=0$, and  write $b=e^2\wedge d+s$ with $d=\sum_{i\geq3}d_ie^i,s=\sum_{i<j, i\geq 3}s_{ij}e^{ij}$, then
                          \begin{align*}
                            e^2\lrcorner(a\wedge b)=-a\wedge d=-e^{46}+e^{57}.
                          \end{align*}
                        It follows from  that
                        \begin{itemize}
                          \item $a_i\neq 0, d_i\neq 0$ for $i=4,5,6,7$,
                          \item $a_4d_5= a_5d_4,\ a_4d_7=a_7d_4,\ a_5d_7-a_7d_5=-1$.
                        \end{itemize}
However, one easily checks that the above two conditions are not compatible.  This means that $a_2$ dose not vanish. Similarly, we can find that $a_i\neq 0$ for $i=3,\cdots,7$.

Step 3: It follows from the equation \eqref{fg} that
\begin{align*}
  \frac{b_{45}}{a_4a_5}-\frac{b_{35}}{a_3a_5}+\frac{b_{34}}{a_3a_4}&=0,\\
  \frac{b_{46}}{a_4a_6}-\frac{b_{36}}{a_3a_6}+\frac{b_{34}}{a_3a_4}&=0,\\
  \frac{b_{56}}{a_5a_6}-\frac{b_{46}}{a_4a_6}+\frac{b_{45}}{a_4a_5}&=0,\\
  \frac{b_{56}}{a_5a_6}-\frac{b_{36}}{a_3a_6}+\frac{b_{35}}{a_3a_5}&=-1.
\end{align*}
One finds that the fist three equalities imply $\frac{b_{56}}{a_5a_6}-\frac{b_{36}}{a_3a_6}+\frac{b_{35}}{a_3a_5}=0$, which contradict with the last one.

In conclusion, nonzero $F_B^2$ does not lie in $\Lambda_{(2)}^4(M)$.
\end{proof}

\begin{lemma}\label{mkl}There do not exist a nonzero 1-form $\alpha\in \Lambda^1(M)$ and a nonzero 2-form $\beta\in \Lambda^2(M)$ such that
$\alpha\wedge \beta\in \Lambda_{(1)}^3(M)\oplus \Lambda_{(2)}^3(M)$.
\end{lemma}
\begin{proof}Assume there exist such forms $\alpha,\beta$. One can always find a point $p\in M$ with $\alpha|_p\neq 0,\beta|_p\neq 0$, then
\begin{align*}
  \alpha|_p\wedge \beta|_p=k\varphi|_p+X\lrcorner(*_\varphi\varphi)|_p,
\end{align*}
where the   constant $k$ and the vector $X\in TM|_p$ cannot vanish simultaneously. Then  we can assume
\begin{itemize}
  \item Case I: $\alpha|_p\wedge \beta|_p= e^{123}+e^{145}+e^{246}+e^{347}-e^{167}+e^{257}-e^{356}+e^{357 }+ e^{256}+ e^{346}- e^{247}$,
  \item Case II: $\alpha|_p\wedge \beta|_p= e^{123}+e^{145}+e^{246}+e^{347}-e^{167}+e^{257}-e^{356}-e^{357 }- e^{256}-e^{346}+ e^{247}$,
\item Case III: $\alpha|_p\wedge \beta|_p= e^{123}+e^{145}+e^{246}+e^{347}-e^{167}+e^{257}-e^{356}$,
  \item Case IV: $\alpha|_p\wedge \beta|_p= e^{357 }+ e^{256}+ e^{346}- e^{247}$.
\end{itemize}

We follow the same arguments as in the proof of Theorem \ref{th}. For the first three cases, write $\alpha|_p=\sum_{i=1}^7a_ie^i$ and  $\beta|_p=\sum_{i<j}b_{ij}e^{ij}$.
To show $a_1\neq0$, we  only need to show there are no $c,d\in \Lambda^1(M)|_p$ such that
\begin{align*}
  c\wedge d=e^{23}+e^{45}-e^{67}.
\end{align*}
Also write $c=\sum_{i\geq 2}c_ie^i, d=\sum_{i\geq 2}d_ie^i$ with nonzero $c_i,d_i$, the above claim follows from that there are  no solutions for the following equations $c_2d_4=c_4d_2, c_3d_4=c_4d_3, c_2d_3-c_3d_2=1$. Similarly,  non-vanishing of $a_2$ is implied by the fact that  there are no $c,d\in \Lambda^1(M)|_p$ such that
\begin{align*}
  c\wedge d=-e^{13}+e^{46}+e^{57}+\delta(e^{56}-e^{47})
\end{align*}
for $\delta=\left\{
              \begin{array}{ll}
                1, & \hbox{Case I,} \\
                -1, & \hbox{Case II,}\\
0,& \hbox{Case III.}
              \end{array}
            \right.
$  Therefore,   all $a_i,i=1,\cdots,7$, are nonzero.
Then we have
\begin{align*}
  \frac{b_{23}}{a_2a_3}-\frac{b_{13}}{a_1a_3}+\frac{b_{12}}{a_1a_2}&=1,\\
  \frac{b_{24}}{a_2a_4}-\frac{b_{14}}{a_1a_4}+\frac{b_{12}}{a_1a_2}&=0,\\
 \frac{b_{34}}{a_3a_4}-\frac{b_{14}}{a_1a_4}+\frac{b_{13}}{a_1a_3}&=0,\\
  \frac{b_{34}}{a_3a_4}-\frac{b_{24}}{a_2a_4}+\frac{b_{23}}{a_2a_3}&=0,
\end{align*}
which admit no solutions. The forth case has been appeared in the proof of Theorem \ref{th}. Consequently, such 1-forms $\alpha,\beta$ satisfying our assumption do not exist.
\end{proof}

Inspired by Theorem \ref{th} and Lemma \ref{mkl}, we introduce  the following definition.
\begin{definition}A 1-form $B\in\Lambda^1(M)$ (as a connection on the trivial $U(1)$-prinicipal bundle P over $M$) is called a special  higher-order $U(1)$-instanton if $B\wedge d_MB\in \Lambda^3_{(3)}(M)$. In particular,  $B$ is called a trivial  special  higher-order $U(1)$-instanton if $B\wedge d_MB=0$.
\end{definition}

\begin{theorem}\label{thh}Assume $B$ is    a special  higher-order $U(1)$-instanton, then
\begin{enumerate}
  \item $B$ is an anti-self-dual higher-order $U(1)$-instanton,
  \item $B$ is a closed 1-form iff $B\wedge *_\varphi d_MB=0$.
\end{enumerate}
\end{theorem}

\begin{proof}(1) By definition,  $B$ satisfies the conditions
\begin{align}
 B\wedge d_MB\wedge\varphi&=0, \label{cz}\\
  *_\varphi (B\wedge d_MB)\wedge\varphi&=0,\label{czz}
\end{align}
The first one immediately implies $d_MB\wedge d_MB\wedge\varphi=0$, and since $\varphi$ is  parallel with respect to the Levi-Civita connection of $g_\varphi$, the second one leads to $d_M^{\dagger_\varphi}(*_\varphi (B\wedge d_MB))\wedge\varphi=*_\varphi(d_B\wedge d_MB)\wedge\varphi=0$, where
$d_M^{\dagger_\varphi}$ is the formal adjoint of $d_M$ with respect to the metric $g_\varphi$. Therefore, $B$ is an anti-self-dual higher-order $U(1)$-instanton.

(2) Decomposing  $d_MB=C_1+C_2$ with $C_i=(d_MB)_{(i)}$, \eqref{cz} is rewritten as
\begin{align*}
  B\wedge d_MB\wedge \varphi&=B\wedge (-2*_\varphi C_1+*_\varphi C_2)\\
&=-3B\wedge *_\varphi C_1+B\wedge *_\varphi d_MB\\
&=0,
\end{align*}
and  \eqref{czz}  is rewritten as
\begin{align*}
B\wedge  d_MB\wedge *_\varphi\varphi=B\wedge C_1\wedge*_\varphi\varphi=0.
\end{align*}
On the other hand, since $C_1=X\lrcorner\varphi$ for some vector field $X$ on $M$, we have the identities
\begin{align*}
 C_1\wedge  *_\varphi\varphi=3*_\varphi X^\vee,\ \  *_\varphi C_1=*_\varphi\varphi\wedge X^\vee,
\end{align*}where $X^\vee$ is the dual 1-form of $X$.
Consequently,
 $B$ is  a special  higher-order $U(1)$-instanton iff $B$ satisfies the conditions
\begin{align*}
  3B \wedge *_\varphi\varphi\wedge X^\vee&=B\wedge*_\varphi d_MB,\\
B\wedge *_\varphi X^\vee&=0,
\end{align*}
Thus, if $B\wedge*_\varphi d_MB=0$, we have
\begin{align}
  B \wedge *_\varphi\varphi\wedge X^\vee&=0,\label{bn}\\
B\wedge *_\varphi X^\vee&=0,\label{bnn}
\end{align}
At a point $p\in M$, we assume $B|_p=e^1, \varphi|_p=\varphi_0 $, then by \eqref{bn},
\begin{align*}
  X^\vee|_p\wedge (e^{14567}+e^{12367}-e^{12345})=0.
\end{align*}
It follows  that  $X^\vee|_p=ce^1$ for some constant $c$, however by \eqref{bnn},  $c$ must be zero, hence  $C_1=0$. Then  Corollary \ref{bnn} implies that  $d_MB=0$.
\end{proof}

\begin{example}
Treat  $\mathbb{R}^7$ as a $G_2$-manifold with the standard fundamental 3-form $\varphi_0$, and consider a 1-form
\begin{align*}
  B=x^2dx^3+ax^4dx^5+bx^6dx^7
\end{align*}
with constant real numbers $a,b$. Then identifying $e^i\simeq dx^i, i=1,\cdots,7$, one easily checks
\begin{align*}
B\wedge dB&=a(x^2e^{345}+x^4e^{235})+b(x^6e^{237}+x^2e^{367})+ab(x^4e^{567}+x^6e^{457}),\\
dB\wedge dB&=2ae^{2345}+2be^{2367}+2abe^{4567},\\
  dB\wedge \varphi_0&=(1+a)e^{12345}+(b-1)e^{12367}+(b-a)e^{14567}\neq0,\\
  dB\wedge *_{\varphi_0}\varphi_0&=(1+a-b)e^{234567},\\
B\wedge dB\wedge\varphi_0&=-(b-a)x^2e^{134567}-a(b-1)x^4e^{123567}-b(1+a)x^6e^{123457},\\
   dB\wedge dB\wedge\varphi_0&=(-a+b+ab)e^{1234567},\\
B\wedge dB\wedge *_{\varphi_0}\varphi_0&=dB\wedge dB\wedge *_{\varphi_0}\varphi_0=0.
\end{align*}
Therefore, we find that
\begin{itemize}
  \item $B$ is an anti-self-dual $U(1)$-instanton iff $b-a=1$,
  \item $B$ is an anti-self-dual  higher-order $U(1)$-instanton iff $-a+b+ab=0$,
  \item $B$ cannot be an anti-self-dual instanton and an anti-self-dual $U(1)$-instanton simultaneously, which can be seen from $dB\neq0$,
\item $B$ is  a special anti-self-dual higher-order $U(1)$-instanton iff $a=b=0$,  equivalently, $B$ is a higher-order flat $U(1)$-instanton or $B$ is  a trivial special anti-self-dual higher-order $U(1)$-instanton.
\end{itemize}

\end{example}

\begin{remark}
Finding (special) anti-self-dual  higher-order $U(1)$-instantons on compact $G_2$-manifolds is an interesting and hard problem.
\end{remark}

\section{Partition function for higher-order $U(1)$-Chern-Simons action}
In this section, we only consider the case of  $\mathcal{P}$ being a trivial $U(1)$-principal bundle over  $\mathcal{M}$, which produces a trivial $U(1)$-principal bundle $P$ over  $M$.
We introduce the ghost fields $\mathfrak{c}, \bar{\mathfrak{c}}$ which lie in $\Lambda^0(M)$ but are fermionic, and $\phi\in \Lambda^0(M)$ which is a Lagrangian multiplier corresponding to  $A$, then under the following transformations
\begin{align*}
  \delta \mathcal{A}_0=-\dot{\mathfrak{c}}:=-\frac{\partial \mathfrak{c}}{\partial t},&\ \ \delta B=-d_M\mathfrak{c},\\
\ \delta\bar{\mathfrak{c}}=\sqrt{-1}\phi,&\ \ \delta\phi=\delta\mathfrak{c}=0,
\end{align*}
the action
\begin{align}
 S_1=S_\textrm{HCS}+\int_L dt\int_Md\mathrm{Vol}_{g_\varphi}(\mathcal{A}_0\phi-\sqrt{-1}\bar{\mathfrak{c}}\dot{\mathfrak{c}}),
\end{align}
where $d\mathrm{Vol}_{g_\varphi}=\sqrt{\det g_\varphi}dX^1\wedge\cdots\wedge dX^7$ is the volume form expressed in terms of local coordinate $\{X^i,i=1,\cdots,7\}$ compatible with the orientation of $M$, is invariant up to a total derivative term
\begin{align*}
  -\int_L dt\int_{M}\frac{\partial}{\partial t}(\mathfrak{c}d_MB\wedge d_MB)\wedge\varphi.
\end{align*}
If one picks $L=\mathbb{R}$, the gauge fixing condition $\mathcal{A}_0=0$ can be imposed. If $L=S^1$, although the above  gauge fixing cannot be reached, one can require \cite{bt}
\begin{align*}
\left\{
  \begin{array}{ll}
\textrm{torus gauge fixing for }  \mathcal{A}_0: \frac{\partial \mathcal{A}_0}{\partial t}=0,  \\
\textrm{constraint on } \phi:  \int_{S^1}dt \phi=0.
  \end{array}
\right.
\end{align*}

In the following, we always assume $M$ is closed and simply-connected,  and $L=S^1$,  then after gauge fixing we consider the following action
\begin{align}
 \mathbf{S}=&-2\int_{S^1} dt\int_M B\wedge d_MB\wedge\dot{B}\wedge \varphi + 3\int_{S^1} dt\int_M \mathcal{A}_0 d_M B\wedge d_MB\wedge \varphi\nonumber\\
&-\sqrt{-1}\int_{S^1} dt\int_Md\mathrm{Vol}_{g_\varphi}\bar{\mathfrak{c}}\dot{\mathfrak{c}}
\end{align}
with the  residual gauge symmetries     given by
\begin{align*}
\delta B =d_Mf,\ \ \delta\mathcal{A}_0=c,\ \ \delta\mathfrak{c}=\delta\mathfrak{c}=0
\end{align*}
for $f\in \Lambda^0(M)$ independent of $t$, $c$ being a constant.
 The partition function is given by path integral as follows
\begin{align}\label{dx}
  Z_\lambda=\frac{1}{\mathrm{Vol}(\mathbf{G})}\int D\mathcal{A}_0 \prod_{n=-\infty}^\infty  DB_n\prod_{n=-\infty}^\infty D\mathfrak{c}_n\prod_{n=-\infty}^\infty  D\bar{\mathfrak{c}}_ne^{\frac{\sqrt{-1}}{\lambda^2}\mathbf{S}},
\end{align}
where $\lambda\in \mathbb{R}$ is the coupling constant, $\mathrm{Vol}(\mathbf{G})$ denotes the formal volume of the group $\mathbb{G}$ consisting of the gauge transformation of the action $\mathbf{S}$.

 By Fourier expansions
\begin{align}
  B&=\sum_{n=-\infty}^{\infty}B_ne^{2\sqrt{-1}\pi nt},\\
  \mathfrak{c}&=\sum_{n=-\infty}^{\infty}\mathfrak{c}_ne^{2\sqrt{-1}\pi nt}
\end{align}
with
 $\bar B_n=B_{-n}$
due to reality of $B$,
we have
\begin{align}
 \mathbf{S}=& \ 4\pi\sqrt{-1}\sum_{n=-\infty}^{\infty}\sum_{m=-\infty}^{\infty}(n+m)\int_MB_n\wedge d_MB_m\wedge B_{-n-m}\wedge \varphi\nonumber\\
& +3\sum_{n=-\infty}^{\infty}\int_M \mathcal{A}_0 d_M B_n\wedge d_MB_{-n}\wedge \varphi+2\pi\sum_{n=-\infty}^{\infty}n\int_Md\mathrm{Vol}_{g_\varphi}\bar{\mathfrak{c}}_{n}\mathfrak{c}_n
\end{align}
Choose a background
\begin{align}
  B^{\textrm{b}}=\sum_{n=-\infty}^\infty B^{\textrm{b}}_ne^{2\sqrt{-1}\pi nt}
\end{align}
  as the critical point of the action $\mathbf{S}$, namely we have
  \begin{align}
    d_MB^{\textrm{b}}\wedge  d_MB^{\textrm{b}}\wedge \varphi&=0,\\
    d_MB^{\textrm{b}}\wedge \dot B^{\textrm{b}}\wedge \varphi&=0,\\
 d_M\mathcal{A}_0\wedge d_MB^{\textrm{b}}\wedge \varphi&=0,
  \end{align}
  or equivalently in terms of Fourier modes
  \begin{align}
   \sum_{m+n=k} d_MB^{\textrm{b}}_m\wedge  d_MB^{\textrm{b}}_n\wedge \varphi&=0,\label{bg}\\
   \sum_{m+n=k}n d_MB_m^{\textrm{b}}\wedge  B_n^{\textrm{b}}\wedge \varphi&=0,\\
 d_M\mathcal{A}_0\wedge d_MB_m^{\textrm{b}}\wedge \varphi&=0,\label{bgg}
  \end{align}
  and express $
   B= B^{\textrm{b}}+\lambda\mathbb{B}
 $
with
 \begin{align}
  \mathbb{B}=\sum_{n=-\infty}^{\infty}\mathbb{B}_ne^{2\sqrt{-1}\pi nt}
\end{align}
  then the  action $\mathbf{S}$ is rewritten as
\begin{align}\label{zss}
  \mathbf{S}=&\ 3\lambda^2[4\pi\sqrt{-1}\sum_{n=-\infty}^{\infty}\sum_{m=-\infty}^{\infty}(n+m)\int_MB_m^{\textrm{b}}\wedge d_M\mathbb{B}_n\wedge \mathbb{B}_{-m-n}\wedge \varphi\nonumber\\
  & +\sum_{n=-\infty}^{\infty}\int_M \mathcal{A}_0 d_M \mathbb{B}_n\wedge d_M\mathbb{B}_{-n}\wedge \varphi] \nonumber \\
&+\lambda^3[4\pi\sqrt{-1}\sum_{n=-\infty}^{\infty}\sum_{m=-\infty}^{\infty}(n+m)\int_M\mathbb{B}_n\wedge d_M\mathbb{B}_m\wedge \mathbb{B}_{-n-m}\wedge \varphi]\nonumber\\
&+2\pi\sum_{n=-\infty}^{\infty}n\int_Md\mathrm{Vol}_{g_\varphi}\bar{\mathfrak{c}}_{n}\mathfrak{c}_n\nonumber\\
=:&\ \lambda^2\mathbf{S}_\textrm{q}+\lambda^3\mathbf{S}_\textrm{int}+\mathbf{S}_\textrm{gh}.
\end{align}
When $\lambda\rightarrow 0$, $Z_\lambda$ can be calculated by semiclassical approximation
\begin{align}\label{mbv}
  Z_{\textrm{sc}}=\frac{1}{\mathrm{Vol}(\mathbf{G})}\int D\mathcal{A}_0 \prod_{n=-\infty}^\infty DB^\textrm{b}_n \prod_{n=-\infty}^\infty D\mathbb{B}_n\prod_{n=-\infty}^\infty D\mathfrak{c}_n\prod_{n=-\infty}^\infty D\bar{\mathfrak{c}}_ne^{\sqrt{-1}(\mathbf{S}_\textrm{q}+\mathbf{S}_\textrm{gh})}.
\end{align}

Firstly, it is clear that
\begin{align*}
\int \prod_{n=-\infty}^\infty D\mathfrak{c}_n\prod_{n=-\infty}^\infty D\bar{\mathfrak{c}}_ne^{\sqrt{-1}\mathbf{S}_\textrm{gh}}={\det}^\prime(\frac{\partial}{\partial t}|_{\Lambda^0(M)\otimes\Lambda^0(S^1)}),
\end{align*}
where the prime above $\det$ means excluding zero mode. By  heat kernel  regularization \cite{bt,li}, we have
\begin{align*}
 & \log{\det}^\prime(\frac{\partial}{\partial t}|_{\Lambda^0(M)\otimes\Lambda^0(S^1)})\nonumber\\
 =&\ {\Tr}^\prime ( e^{-\varepsilon\Delta^{(0)}_M}\log\frac{\partial}{\partial t}|_{\Lambda^0(M)\otimes\Lambda^0(S^1)}) \\ =&\ \log{\det}^\prime(\frac{\partial}{\partial t}|_{\Lambda^0(S^1)}),
\end{align*}
where $\Delta_M^{(0)}$ is the Laplacian on $\Lambda^0(M)$, therefore we get
\begin{align}
&\int \prod_{n=-\infty}^\infty D\mathfrak{c}_n\prod_{n=-\infty}^\infty D\bar{\mathfrak{c}}_ne^{\sqrt{-1}\mathbf{S}_\textrm{gh}}\nonumber\\
=&\ {\det}^\prime(\frac{\partial}{\partial t}|_{\Lambda^0(S^1)})
 =\prod_{n>0}(2\pi n)^2=(\frac{\sqrt{2\pi}}{\sqrt{2\pi}})^2=1.
\end{align}

Next we deal with the path integral
\begin{align}\label{mkc}
&\frac{1}{\mathrm{Vol}(\mathbf{G})}\int D\mathcal{A}_0D B^bD\mathbb{B}e^{\sqrt{-1}\mathbf{S}_\textrm{q}}\nonumber\\
=&\ \frac{1}{\mathrm{Vol}(\mathbf{G})}\int D\mathcal{A}_0 D B^b D\mathbb{B}e^{3\sqrt{-1}\int_{S^1}dt\int_M\mathbb{B}\wedge d_M\circ(2B^b\wedge \varphi\wedge \frac{\partial}{\partial t}+d_M\mathcal{A}_0\wedge\varphi\wedge)\mathbb{B}}\nonumber\\
=&\ \frac{1}{\mathrm{Vol}(\mathbf{G})}\int D\mathcal{A}_0 D \prod_{n=-\infty}^\infty B_n^b\prod_{n=-\infty}^\infty D\mathbb{B}_n \nonumber\\ &\ \ \ \ \ \ \ \ \ \ \ \ \ \ \ \ \ \ \ \{e^{3\sqrt{-1}\sum_{n=-\infty}^{\infty}[-4\pi\sqrt{-1}\sum_{m=-\infty}^{\infty}(n+m)\int_M\mathbb{B}_n \wedge d_M\circ( B_m^b\wedge \varphi\wedge) \mathbb{B}_{-m-n}+ \mathbb{B}_n \wedge(d_M\mathcal{A}_0 \wedge \varphi\wedge d_M)\mathbb{B}_{-n}]}\}.
  \end{align}
One again  applies  heat kernel regularization, namely inserts the fact  $e^{\varepsilon\frac{\partial^2}{\partial t^2}}$ in  the operator
\begin{align*}
\mathcal{D}_{\mathcal{A}_0,B^b,\varphi}= d_M\circ(2B^b\wedge \varphi\wedge \frac{\partial}{\partial t}+d_M\mathcal{A}_0\wedge\varphi\wedge)
\end{align*}
as
\begin{align}
  \mathcal{D}_{\mathcal{A}_0,B^b,\varphi}\rightarrow d_M\circ(2e^{\varepsilon\frac{\partial^2}{\partial t^2}}B^b\wedge \varphi\wedge e^{\varepsilon\frac{\partial^2}{\partial t^2}}\frac{\partial}{\partial t}+d_M\mathcal{A}_0\wedge\varphi\wedge e^{\varepsilon\frac{\partial^2}{\partial t^2}})
\end{align}
then only  zero modes of $ B^b$ and $\mathbb{B}$ survive in \eqref{mkc}.
Therefore, we only need to consider
\begin{align}
\frac{1}{\mathrm{Vol}(\mathbf{G})}\int D B^b  D\mathbb{B}e^{\sqrt{-1}\mathbf{S}_\textrm{q}}&=\frac{1}{\mathrm{Vol}(\mathbf{G})}\int DB_0^b D\mathbb{B}_0 e^{3\sqrt{-1}\int_M  (\mathbb{B}_0\wedge d_M\mathcal{A}_0)d_M(\mathbb{B}_0\wedge\varphi)}.
\end{align}
We will use Schwarz-Pestun-Witten's method \cite{sc,pw} to formally calculate it.

Introduce the following notations.
\begin{itemize}
  \item $\Delta_\varphi=d^{\dagger_\varphi}_Md_M+d_M d^{\dagger_\varphi}_M$, $\Delta'_\varphi=d^{\dagger_\varphi}_Md_M$, $\Delta''_\varphi=d_M d^{\dagger_\varphi}_M$,
\item $\Lambda_{\textrm{closed}}^\bullet(M)$: the space of closed $\bullet$-forms,
  \item $\Lambda^2_{\mathcal{A}_0}(M)=\{\beta\in\Lambda^2(M):\beta=\alpha\wedge d_M\mathcal{A}_0\textrm{ for some } \alpha\in\Lambda^1(M)\}\simeq\Lambda^1(M)$,
  \item $\widetilde{\Lambda}^2_{\mathcal{A}_0}(M)=\{\beta\in \Lambda^2_{\mathcal{A}_0}(M): \Delta'_\varphi \beta\in \Lambda^2_{\mathcal{A}_0}(M)\}$,
\item $E\Lambda^2_{\mathcal{A}_0}(M)$: the space of exact 2-forms in $\Lambda^2_{\mathcal{A}_0}(M)$,
\item  $\Lambda^1_{\mathcal{A}_0}(M)=\{\alpha\in\Lambda^1(M):d_M\alpha\in \Lambda^2_{\mathcal{A}_0}(M)\}$,
\item $\widetilde{\Lambda}^1_{\mathcal{A}_0}(M)=\{\alpha\in \Lambda^1_{\mathcal{A}_0}(M): \Delta'_\varphi \alpha\in \Lambda^1_{\mathcal{A}_0}(M)\}$,
\item $\Lambda^2_\varphi(M)=\{\beta\in \Lambda^2(M): \beta\wedge \varphi=0\}$,
\item $E\Lambda^2_\varphi(M)$: the space of exact 2-forms in $\Lambda^2_{\varphi}(M)$,
\item  $\Lambda^1_{\varphi}(M)=\{\alpha\in\Lambda^1(M):d_M\alpha\in \Lambda^2_{\varphi}(M)\}$,
\end{itemize}
Then  formally, we have
\begin{align}
  &\frac{1}{\mathrm{Vol}(\mathbf{G})}\int  D\mathbb{B}_0 e^{3\sqrt{-1}\int_M  (\mathbb{B}_0\wedge d_M\mathcal{A}_0)d_M(\mathbb{B}_0\wedge\varphi)}\nonumber\\
=&\ \frac{1}{\mathrm{Vol}(\mathbf{G})}\{[{\det}'(9\Delta'_\varphi|_{\widetilde{\Lambda}^2_{\mathcal{A}_0}(M)})]^{-\frac{1}{4}}
[{\det}'(9\Delta'_\varphi|_{\Lambda^4_{(2)}(M)})]^{-\frac{1}{4}}\Vol(E\Lambda^2_{\mathcal{A}_0}(M))
\Vol(E\Lambda^2_{\varphi}(M))\}^{\frac{1}{2}}.
\end{align}
Moreover,  we have
\begin{align}
 &\Vol(E\Lambda^2_{\mathcal{A}_0}(M))\Vol(E\Lambda^2_{\varphi}(M))\nonumber\\
=&\ \frac{ \Vol(\Lambda^1_{\mathcal{A}_0}(M))}{\Vol(\Lambda^1_{\textrm{closed}}(M))}\frac{ \Vol(\Lambda^1_{\varphi}(M))}{\Vol(\Lambda^1_{\textrm{closed}}(M))}[{\det}'(\Delta'_\varphi|_{\widetilde{\Lambda}^1_{\mathcal{A}_0}(M)})]^{\frac{1}{2}}[{\det}'(\Delta'_\varphi|_{{\Lambda}^1_{\varphi}(M)})]^{\frac{1}{2}}\nonumber\\
=&\ \frac{ \Vol(\Lambda^1_{\mathcal{A}_0}(M))\Vol(\Lambda^1_{\varphi}(M))}{[\Vol(\Lambda^0(M))]^2}[\frac{\Vol(H^0(M,\mathbb{R}))}{\Vol(H^1(M,\mathbb{R}))}]^2\frac{[{\det}'(\Delta'_\varphi|_{\widetilde{\Lambda}^1_{\mathcal{A}_0}(M)})]^{\frac{1}{2}}[{\det}'(\Delta'_\varphi|_{{\Lambda}^1_{\varphi}(M)})]^{\frac{1}{2}}}
{{\det}'(\Delta_\varphi|_{{\Lambda}^0(M)})}.
\end{align}
Let $\mathbf{G}'\subset \mathbf{G}$ be the subgroup of $\mathbf{G}$ consisting of gauge transformations on $\mathbb{B}_0$ (preserving $\mathcal{A}_0,B_0^b$), then
 $\mathrm{Vol}(\mathbf{G}')$ can be renormalized to be $\frac{ [\Vol(\Lambda^1_{\mathcal{A}_0}(M))\Vol(\Lambda^1_{\varphi}(M))]^{\frac{1}{2}}}{\Vol(\Lambda^0(M))}$. Also, $\frac{\Vol(H^0(M,\mathbb{R}))}{\Vol(H^1(M,\mathbb{R}))}$ can be renormalized to be 1.
Consequently, we arrive at
\begin{align}
 &\frac{1}{\mathrm{Vol}(\mathbf{G})}\int D\mathcal{A}_0D B_0^bD\mathbb{B}e^{\sqrt{-1}\mathbf{S}_\textrm{q}}\nonumber\\
 =& \int_{ \mathbb{M}^b} D\mathcal{A}_0D B_0^b\frac{[{\det}'(\Delta'_\varphi|_{\widetilde{\Lambda}^1_{\mathcal{A}_0}(M)})]^{\frac{1}{4}}[{\det}'(\Delta'_\varphi|_{{\Lambda}^1_{\varphi}(M)})]^{\frac{1}{4}}}
{[{\det}'(\Delta_\varphi|_{{\Lambda}^0(M)})]^{\frac{1}{2}}[{\det}'(9\Delta'_\varphi|_{\widetilde{\Lambda}^2_{\mathcal{A}_0}(M)})]^{\frac{1}{8}}
[{\det}'(9\Delta'_\varphi|_{\Lambda^4_{(2)}(M)})]^{\frac{1}{8}}},
\end{align}
where $\mathbb{M}^b$
 is the formal  the moduli space  defined as
\begin{align*}
  \mathbb{M}^b=\frac{\{(B^b_0,\mathcal{A}_0)\in\Lambda^1(M)\oplus\Lambda^0(M):d_MB^b_0\wedge d_MB^b_0\wedge \varphi=d_M\mathcal{A}_0\wedge d_MB^b_0\wedge \varphi=0\}}{\textrm{ gauge transformations }}.
\end{align*}
In addition, we have
\begin{align*}
{\det}'(9\Delta'_\varphi|_{\Lambda^4_{(2)}(M)})&={\det}'(9\Delta'_\varphi|_{\Lambda^1(M)})\\
&=\frac{{\det}'(9\Delta_\varphi|_{\Lambda^1(M)})}{{\det}'(9\Delta''_\varphi|_{\Lambda^1(M)})}=\frac{{\det}'(9\Delta_\varphi|_{\Lambda^1(M)})}{{\det}'(9\Delta_\varphi|_{\Lambda^0(M)})}\\
&=9^{\dim H^0(M,\mathbb{R})-\dim H^1(M,\mathbb{R})}\frac{{\det}'(\Delta_\varphi|_{\Lambda^1(M)})}{{\det}'(\Delta_\varphi|_{\Lambda^0(M)})},
\end{align*}
where  the third equality is due to  the formula \cite{je,mc}
\begin{align*}
  {\det}'(c\Delta_\varphi|_{\Lambda^p(M)})=c^{-\dim H^p(M,\mathbb{R})}{\det}'(\Delta_\varphi|_{\Lambda^p(M)})
\end{align*}
for some constant $c$.
 Finally, we obtain
\begin{align}\
 Z_{\textrm{sc}}=\frac{1}{9}\frac{[{\det}'(\Delta'_\varphi|_{{\Lambda}^1_{\varphi}(M)})]^{\frac{1}{4}}}{[{\det}'(\Delta_\varphi|_{\Lambda^1(M)})]^{\frac{1}{8}}[{\det}'(\Delta_\varphi|_{\Lambda^0(M)})]^{\frac{3}{8}}}\int_{ \mathbb{M}^b} D\mathcal{A}_0D B_0^b\frac{[{\det}'(\Delta'_\varphi|_{\widetilde{\Lambda}^1_{\mathcal{A}_0}(M)})]^{\frac{1}{4}}}
{[{\det}'(9\Delta'_\varphi|_{\widetilde{\Lambda}^2_{\mathcal{A}_0}(M)})]^{\frac{1}{8}}
}.
\end{align}

\section{Higher-order $U(1)$-Chern-Simons actions for nontrivial $U(1)$-principal bundles}
In the final section, we discuss the higher-order $U(1)$-Chern-Simons actions for nontrivial $U(1)$-principal bundles. As mentioned in Introduction, we need the
Deligne-Beilinson cohomology theory. Hence, we first recall some basis materials in this theory. The original Deligne-Beilinson cohomology is defined for algebraic varieties, and the smooth analogy of this theory is also called Cheeger-Simons cohomology \cite{dbb}.  For a compact manifold $X$, the Deligne complex of sheaves is given by
\begin{align*}
\mathbf{DC}_{R(\ell)}:  0\rightarrow R(\ell)\rightarrow\Lambda^0(X,\ell)\xrightarrow{d}\Lambda^1(X,\ell)\xrightarrow{d}\cdots\xrightarrow{d}\Lambda^{\ell}(X,\ell),
\end{align*}
where $\Lambda^\bullet(X,\ell)=(2\pi\sqrt{-1})^{\ell}\Lambda^\bullet(X)$,  $R(\ell)=(2\pi\sqrt{-1})^{\ell}R$ for a subring $R$ of $\mathbb{R}$ and some integer $\ell\geq0$,
then the $q$-order Deligne-Beilinson cohomology group $H^q_{\mathrm{DB}}(X,R(\ell))$ is defined as the $q$-th  hypercohomology of the Deligne complex $\mathbf{DC}_{R(\ell)}$, i.e. $H^q_{\mathrm{DB}}(X,R(\ell))=\mathbb{H}^q(\mathbf{DC}_{R(\ell)})$.
Taking  an open cover $\mathcal{U}=\{\mathcal{U}_\alpha\}_{\alpha\in I}$ of $X$, i.e. $X=\bigcup_{\alpha\in I} \mathcal{U}_\alpha$, we consider the \v{C}ech  resolution of  $\mathbf{DC}_{R(\ell)}$:
\begin{equation*}
  \CD
  R(\ell) @>\iota>> \Lambda^0(X,\ell)@>d>> \Lambda^1(X,\ell)@>d>>\cdots @>d>>\Lambda^{\ell}(X,\ell) \\
  @V\iota VV @V \iota VV @V\iota VV@ V\iota VV @V\iota VV \\
 \mathcal{C}^0(\mathcal{U}, R(\ell)) @>\iota>>  \mathcal{C}^0(\mathcal{U}, \Lambda^0(X,\ell))@>d_M>>  \mathcal{C}^0(\mathcal{U},\Lambda^1(X,\ell))@>d>>\cdots @>d>> \mathcal{C}^0(\mathcal{U},\Lambda^{\ell}(X,\ell))\\
@V\delta VV @V \delta VV @V\delta VV@ V\delta VV @V\delta VV \\
\mathcal{C}^1(\mathcal{U}, R(\ell)) @>\iota>>  \mathcal{C}^1(\mathcal{U}, \Lambda^0(X,\ell))@>d>>  \mathcal{C}^1(\mathcal{U},\Lambda^1(X,\ell))@>d>>\cdots @>d>> \mathcal{C}^1(\mathcal{U},\Lambda^{\ell}(X,\ell))\\
@V\delta VV @V \delta VV @V\delta VV@ V\delta VV @V\delta VV \\
\vdots @>\iota>>  \vdots @>d>>  \vdots @>d>>\vdots @>d>> \vdots
\endCD,
\end{equation*}
where $\iota$ denotes the natural embedding, $\delta$ denotes the \v{C}ech operator, and $\mathcal{C}^q(\mathcal{U},\mathcal{F})$ denotes
the space of  $q$-dimensional \v{C}ech cochains  for a sheaf $\mathcal{F}$ over $M$. Therefore
\begin{align*}
 H^q_{\textrm{DB}}(X,R(\ell))=\lim\limits_{\substack{\longrightarrow\\\mathcal{U}}}H^q(\mathrm{Tot}_\mathcal{U}(\mathbf{DC}_{R(\ell)})),
\end{align*}
 where $\mathrm{Tot}_\mathcal{U}(\mathbf{DC}_{R(\ell)})$ is the total complex of the  \v{C}ech  resolution of  $\mathbf{DC}_{R(\ell)}$ associated to the cover $\mathcal{U}$. In particular, if $\mathcal{U}$ is a simple (good) cover \cite{gt}, we have $H^q_{\textrm{DB}}(M,R(\ell))=H^q(\mathrm{Tot}_\mathcal{U}(\mathbf{DC}_{R(\ell)}))$ \cite{bg}.

For our purpose, we take $R=\mathbb{Z}$. From the following commutative diagram
\begin{align*}
 \CD
  \mathbb{Z}(\ell) @>>> \Lambda^0(X,\ell) @>d>> \Lambda^1(X,\ell)@>d>>\cdots @>d>> \Lambda^{\ell}(X,\ell)\\
  @V  VV @V \exp_\ell VV @V\frac{1}{(2\pi\sqrt{-1})^{\ell-1}} VV @V\frac{1}{(2\pi\sqrt{-1})^{\ell-1}} VV@V\frac{1}{(2\pi\sqrt{-1})^{\ell-1}} VV\\
  0 @>>>  \underline{U(1)}_X@> d\log>> \Lambda^1(X,1)@>d>>\cdots @>d>> \Lambda^{\ell}(M,1)
\endCD,
\end{align*}
 where $\underline{U(1)}_X$ denotes sheaf of $U(1)$-valued functions over $X$, and $\exp_\ell(f)=e^{\frac{f}{(2\pi\sqrt{-1})^{\ell-1}}}$, we find that
 \begin{align*}
                                                                               H^q_{\mathrm{DB}}(X,\mathbb{Z}(\ell))\simeq \mathbb{H}^{q-1}(\mathbf{\underline{U(1)}}_X(\ell))
                                                                              \end{align*}
where $\mathbf{\underline{U(1)}}_X(\ell)$ denote the complex $\underline{U(1)}_X\xrightarrow{d\log}\Lambda^1(X,1)\xrightarrow{d}\cdots\xrightarrow{d}\Lambda^{\ell-1}(X,1)$. It immediately implies that Deligne-Beilinson cohomology  $H^2_{\mathrm{DB}}(X,\mathbb{Z}(1))$ parameterize the isomorphism classes of $U(1)$-principal bundles with connections over $M$ \cite{bg,gt}. More explicitly, choosing a simple cover $\mathcal{U}=\{\mathcal{U}_\alpha\}_{\alpha\in I}$ of $X$ with index set $I$,
an element in $H^2_{\mathrm{DB}}(X,\mathbb{Z}(1))$ is represented  by  a triple $(\{\mathbb{A}_\alpha\}_{\alpha\in I}, \{\Gamma_{\alpha\beta}\}_{\alpha,\beta\in I},\{\Upsilon_{\alpha\beta\gamma}\}_{\alpha,\beta,\gamma\in I})$ of \v{C}ech cochains
 satisfying
\begin{align}
  \mathbb{A}_\beta-\mathbb{A}_\alpha&=d\Gamma_{\alpha\beta}\label{41}\\
\Gamma_{\beta\gamma}-\Gamma_{\alpha\gamma}+\Gamma_{\alpha\beta}&=\Upsilon_{\alpha\beta\gamma}\label{42}\\
\Upsilon_{\beta\gamma\delta}-\Upsilon_{\alpha\gamma\delta}+\Upsilon_{\alpha\beta\delta}-\Upsilon_{\alpha\beta\gamma}&=0,
\end{align}
where $\{\mathbb{A}_\alpha\}\in \mathcal{C}^0(\mathcal{U},\Lambda^1(X,1)), \{\Gamma_{\alpha\beta}\}\in \mathcal{C}^1(\mathcal{U},\Lambda^0(X,1)), \{\Upsilon_{\alpha\beta\gamma}\}\in \mathcal{C}^0(\mathcal{U},\mathbb{Z}(1))$.
Writing  $g_{\alpha\beta}=e^{-\Gamma_{\alpha\beta}}, \Gamma_{\alpha\beta}=-2\pi\sqrt{-1}\Lambda_{\alpha\beta}$, then \eqref{42} indicates that  $\{g_{\alpha\beta}\}$ defines  transition functions on a $U(1)$-principal bundle, and \eqref{41} means that  $\{\mathbb{A}_\alpha\}$ defines a connection on such $U(1)$-principal bundle.

Now $X=\mathcal{M}=M\times S^1$, where $M$  is a closed $G_2$-manifold whose fundamental 3-form $\varphi$ is assumed to be an integral cohomology class.  One views $\varphi\in H^3(M,\mathbb{Z})$
as an element of Deligne-Beilinson cohomology $H^3_{\mathrm{DB}}(M,\mathbb{Z}(1))$, hence an element of  $H^3_{\mathrm{DB}}(\mathcal{M},\mathbb{Z}(1))$ by natural inclusion.  Let $\mathcal{A},\mathcal{B},\mathcal{C}\in H^2_{\mathrm{DB}}(\mathcal{M},\mathbb{Z}(1))$ describe isomorphism classes of $U(1)$-principal bundles with connections over $\mathcal{M}$.
Taking  a simple cover $\mathcal{U}=\{\mathcal{U}_\alpha\}_{\alpha\in I}$ of $\mathcal{M}$,   we represent $\mathcal{A},\mathcal{B},\mathcal{C}\in H^2_{\mathrm{DB}}(\mathcal{M},\mathbb{Z}(1))$ as
\begin{align*}
 \mathcal{A}&=(\{\mathbb{A}_\alpha\}, \{\Gamma_{\alpha\beta}\},\{\Upsilon_{\alpha\beta\gamma}\}),\\
\mathcal{B}&=(\{\mathbb{B}_\alpha\}, \{\Theta_{\alpha\beta}\},\{\Lambda_{\alpha\beta\gamma}\}),\\
\mathcal{C}&=(\{\mathbb{C}_\alpha\}, \{\Psi_{\alpha\beta}\},\{\Omega_{\alpha\beta\gamma}\}),\end{align*}
and also represent  $\varphi\in H^3_{\mathrm{DB}}(\mathcal{M},\mathbb{Z}(1))$  as
\begin{align*}
  \varphi&=(  \{\chi_{\alpha\beta}\},\{\tau_{\alpha\beta\gamma}\},  \{2\pi\sqrt{-1}\theta_{\alpha\beta\gamma\eta}\}),
\end{align*}
where $\{2\pi\sqrt{-1}\theta_{\alpha\beta\gamma\eta}\}\in C^3(\mathcal{U},\mathbb{Z}(1))$ are determined via the isomorphism
\begin{align*}
 H^3(\mathcal{M},\mathbb{Z})&\simeq H_{\textrm{\v{C}ech}}^3(\mathcal{M},\mathbb{Z}),\\
 \varphi&\mapsto \{\theta_{\alpha\beta\gamma\eta}\},
\end{align*}
 and $\{\chi_{\alpha\beta}\}\in C^1(\mathcal{U},\Lambda^1(\mathcal{M},1)), \{\tau_{\alpha\beta\gamma}\}\in C^2(\mathcal{U},\Lambda^0(\mathcal{M},1))$ are determined as follows
\begin{align*}
  \chi_{\alpha\beta}+\sum_{\gamma\in I}(d\tau_{\alpha\beta\gamma})\xi_\gamma&=0,\\
\tau_{\alpha\beta\gamma}-2\pi\sqrt{-1}\sum_{\eta\in I}\theta_{\alpha\beta\gamma\eta}\xi_\eta&=0
\end{align*}
for a partition$\{\xi_\alpha\}_{\alpha\in I}$ of unity  subordinate to the simple cover $\mathcal{U}$.

Consider  a polyhedral decomposition $(\{\mathfrak{P}^{(8)}_{\alpha_0}\}_{\alpha_0\in I},\cdots, \{\mathfrak{P}^{(0)}_{\alpha_0\cdots\alpha_8}\}_{\alpha_0,\cdots, \alpha_8\in I})$ of $M$ subordinate to $\mathcal{U}$, where $\mathfrak{P}^{(d)}_{\alpha_0\cdots\alpha_{8-d}}$ is a $d$-dimensional submanifold of $M$ lying in $\mathcal{U}_{\alpha_0}\bigcap\cdots \mathcal{U}_{\alpha_{8-d}}$, then we construct the gauge invariant $U(1)$-ABC action.
We begin with the following action
\begin{align*}
  I_0=\sum_{\alpha_0,\alpha_1,\alpha_2,\alpha_3\in I}\int_{\mathfrak{P}^{(5)}_{\alpha_0\alpha_1\alpha_2\alpha_3}}\theta_{\alpha_0\alpha_1\alpha_2\alpha_3}\mathbb{A}_{\alpha_0}\wedge d\mathbb{B}_{\alpha_0}\wedge d\mathbb{C}_{\alpha_0}.
\end{align*}
Under the lacal gauge transformation
\begin{align*}
  \mathbb{A}_\alpha\mapsto\mathbb{A}_\alpha+d\mathfrak{a}_\alpha, \   \ \mathbb{B}_\alpha\mapsto\mathbb{B}_\alpha+d\mathfrak{b}_\alpha,  \ \ \mathbb{C}_\alpha\mapsto\mathbb{C}_\alpha+d\mathfrak{c}_\alpha,
\end{align*}
for $\mathfrak{a}_\alpha,\mathfrak{b}_\alpha,\mathfrak{c}_\alpha\in \mathcal{C}^0(\mathcal{U},\Lambda^1(\mathcal{M},1)$, its variation is  given by
\begin{align*}
  \Delta_\textrm{lol} I_0&=\sum_{\alpha_0,\alpha_1,\alpha_2,\alpha_3,\alpha_4\in I}\int_{\mathfrak{P}^{(4)}_{\alpha_4\alpha_0\alpha_1\alpha_2\alpha_3}-\mathfrak{P}^{(4)}_{\alpha_0\alpha_4\alpha_1\alpha_2\alpha_3}+\mathfrak{P}^{(4)}_{\alpha_0\alpha_1\alpha_4\alpha_2\alpha_3}-
\mathfrak{P}^{(4)}_{\alpha_0\alpha_1\alpha_2\alpha_4\alpha_3}+\mathfrak{P}^{(4)}_{\alpha_0\alpha_1\alpha_2\alpha_3\alpha_4}}\\
&\ \ \ \ \ \ \ \  \ \ \ \ \ \  \ \ \ \ \ \  \ \ \ \ \ \  \ \ \ \ \ \  \ \  [\theta_{\alpha_0\alpha_1\alpha_2\alpha_3}\mathfrak{a}_{\alpha_0}\wedge d\mathbb{B}_{\alpha_0}\wedge d\mathbb{C}_{\alpha_0}]\\
&=\sum_{\alpha_0,\alpha_1,\alpha_2,\alpha_3,\alpha_4\in I}\int_{\mathfrak{P}^{(4)}_{\alpha_0\alpha_1\alpha_2\alpha_3\alpha_4}}[(\theta_{\alpha_0\alpha_1\alpha_2\alpha_3}-\theta_{\alpha_0\alpha_1\alpha_2\alpha_4}+\theta_{\alpha_0\alpha_1\alpha_3\alpha_4}-\theta_{\alpha_0\alpha_2\alpha_3\alpha_4})\\
&\ \ \ \ \ \ \ \  \ \ \ \ \ \  \ \ \ \ \ \  \ \ \ \ \ \  \ \ \ \ \ \  \ \cdot\mathfrak{a}_{\alpha_0}\wedge d\mathbb{B}_{\alpha_0}\wedge d\mathbb{C}_{\alpha_0}+\theta_{\alpha_1\alpha_2\alpha_3\alpha_4}\mathfrak{a}_{\alpha_1}d\mathbb{B}_{\alpha_1}\wedge d\mathbb{C}_{\alpha_1}]\\
&=\sum_{\alpha_0,\alpha_1,\alpha_2,\alpha_3,\alpha_4\in I}\int_{\mathfrak{P}^{(4)}_{\alpha_0\alpha_1\alpha_2\alpha_3\alpha_4}}\theta_{\alpha_1\alpha_2\alpha_3\alpha_4}(\mathfrak{a}_{\alpha_1}-\mathfrak{a}_{\alpha_0})d\mathbb{B}_{\alpha_0}\wedge d\mathbb{C}_{\alpha_0}],
\end{align*}
 which can be eliminated by the variation of the action
 \begin{align*}
   I_1=-\sum_{\alpha_0,\alpha_1,\alpha_2\alpha_3,\alpha_4\in I}\int_{\mathfrak{P}^{(4)}_{\alpha_0\alpha_1\alpha_2\alpha_3\alpha_4}}\theta_{\alpha_1\alpha_2\alpha_3\alpha_4}\Gamma_{\alpha_0\alpha_1}d\mathbb{B}_{\alpha_0}\wedge d\mathbb{C}_{\alpha_0}
 \end{align*}
under the transformation $\Gamma_{\alpha\beta}\mapsto\Gamma_{\alpha\beta}+\mathfrak{a}_\beta-\mathfrak{a}_{\alpha}$. However, the action $I_0+I_1$ is not invariant under the large gauge transformation
$\Gamma_{\alpha\beta}\mapsto\Gamma_{\alpha\beta}+z_{\alpha\beta}$ for $z_{\alpha\beta}\in\mathcal{C}^1(\mathcal{U},\mathbb{Z}(1))$. Indeed,  the variation of $I_1$ is given by \begin{align*}
                             \Delta_\textrm{lar}I_1= &\ \Delta^{1ar}_\textrm{l}I_1+ \Delta_\textrm{lar}^{2}I_1\\
=& -\sum_{\alpha_0,\alpha_1,\alpha_2,\alpha_3,\alpha_4,\alpha_5\in I}\int_{\mathfrak{P}^{(3)}_{\alpha_0\alpha_1\alpha_2\alpha_3\alpha_4\alpha_5}}\theta_{\alpha_2\alpha_3\alpha_4\alpha_5}(\delta z)_{\alpha_0\alpha_1\alpha_2}\mathbb{B}_0d\mathbb{C}_{\alpha_0} \\& - \sum_{\alpha_0,\alpha_1,\alpha_2,\alpha_3,\alpha_4,\alpha_5\in I} \int_{\mathfrak{P}^{(3)}_{\alpha_0\alpha_1\alpha_2\alpha_3\alpha_4\alpha_5}}\theta_{\alpha_2\alpha_3\alpha_4\alpha_5}z_{\alpha_1\alpha_2}(\mathbb{B}_{\alpha_1}d\mathbb{C}_{\alpha_1}-\mathbb{B}_{\alpha_0}d\mathbb{C}_{\alpha_0}).
                            \end{align*}
$\Delta^{1}_\textrm{lar}I_1$ can be eliminated by variation of the
 action\begin{align*}
         I_2=\sum_{\alpha_0,\alpha_1,\alpha_2,\alpha_3,\alpha_4,\alpha_5\in I}\int_{\mathfrak{P}^{(3)}_{\alpha_0\alpha_1\alpha_2\alpha_3\alpha_4\alpha_5}}\theta_{\alpha_2\alpha_3\alpha_4\alpha_5}\Upsilon_{\alpha_0\alpha_1\alpha_2}\mathbb{B}_{\alpha_0}\wedge d\mathbb{C}_{\alpha_0}
       \end{align*}
under the transformation $\Upsilon_{\alpha\beta\gamma}\mapsto \Upsilon_{\alpha\beta\gamma}+(\delta z)_{\alpha\beta\gamma}$.
As well as, $\Delta_\textrm{lar}^{2}I_1$ can be calculated as
\begin{align*}
 \Delta_\textrm{lar}^{2}I_1=&-\sum_{\alpha_0,\alpha_1,\alpha_2,\alpha_3,\alpha_4,\alpha_5\in I} \int_{\mathfrak{P}^{(3)}_{\alpha_0\alpha_1\alpha_2\alpha_3\alpha_4\alpha_5}}\theta_{\alpha_2\alpha_3\alpha_4\alpha_5}z_{\alpha_1\alpha_2}d(\Theta_{\alpha_0\alpha_1}d\mathbb{C}_{\alpha_0})\\
=& \sum_{\alpha_0,\alpha_1,\alpha_2,\alpha_3,\alpha_4,\alpha_5,\alpha_6\in I}\int_{\mathfrak{P}^{(2)}_{\alpha_0\alpha_1\alpha_2\alpha_3\alpha_4\alpha_5\alpha_6}}\theta_{\alpha_3\alpha_4\alpha_5\alpha_6}(\delta z)_{\alpha_1\alpha_2\alpha_3}\Theta_{\alpha_0\alpha_1}d\mathbb{C}_{\alpha_0}\\
&-\sum_{\alpha_0,\alpha_1,\alpha_2,\alpha_3,\alpha_4,\alpha_5,\alpha_6\in I}\int_{\mathfrak{P}^{(2)}_{\alpha_0\alpha_1\alpha_2\alpha_3\alpha_4\alpha_5\alpha_6}}\theta_{\alpha_3\alpha_4\alpha_5\alpha_6} z_{\alpha_2\alpha_3}\Lambda_{\alpha_0\alpha_1\alpha_2}d\mathbb{C}_{\alpha_0}\\
=&\  \sum_{\alpha_0,\alpha_1,\alpha_2,\alpha_3,\alpha_4,\alpha_5,\alpha_6\in I}\int_{\mathfrak{P}^{(2)}_{\alpha_0\alpha_1\alpha_2\alpha_3\alpha_4\alpha_5\alpha_6}}\theta_{\alpha_3\alpha_4\alpha_5\alpha_6}(\delta z)_{\alpha_1\alpha_2\alpha_3}\Theta_{\alpha_0\alpha_1}d\mathbb{C}_{\alpha_0}\\
&-\sum_{\alpha_0,\alpha_1,\alpha_2,\alpha_3,\alpha_4,\alpha_5,\alpha_6,\alpha_7\in I}\int_{\mathfrak{P}^{(1)}_{\alpha_0\alpha_1\alpha_2\alpha_3\alpha_4\alpha_5\alpha_6\alpha_7}}\theta_{\alpha_4\alpha_5\alpha_6\alpha_7} (\delta z)_{\alpha_2\alpha_3\alpha_4}\Lambda_{\alpha_0\alpha_1\alpha_2}\mathbb{C}_{\alpha_0}\\
&-\sum_{\alpha_0,\alpha_1,\alpha_2,\alpha_3,\alpha_4,\alpha_5,\alpha_6,\alpha_7\in I}\int_{\mathfrak{P}^{(1)}_{\alpha_0\alpha_1\alpha_2\alpha_3\alpha_4\alpha_5\alpha_6\alpha_7}}\theta_{\alpha_4\alpha_5\alpha_6\alpha_7}  z_{\alpha_3\alpha_4}\Lambda_{\alpha_1\alpha_2\alpha_3}d\Psi_{\alpha_0\alpha_1}\\
\approx&  \sum_{\alpha_0,\alpha_1,\alpha_2,\alpha_3,\alpha_4,\alpha_5,\alpha_6\in I}\int_{\mathfrak{P}^{(2)}_{\alpha_0\alpha_1\alpha_2\alpha_3\alpha_4\alpha_5\alpha_6}}\theta_{\alpha_3\alpha_4\alpha_5\alpha_6}(\delta z)_{\alpha_1\alpha_2\alpha_3}\Theta_{\alpha_0\alpha_1}d\mathbb{C}_{\alpha_0}\\
&-\sum_{\alpha_0,\alpha_1,\alpha_2,\alpha_3,\alpha_4,\alpha_5,\beta_6,\alpha_7\in I}\int_{\mathfrak{P}^{(1)}_{\alpha_0\alpha_1\alpha_2\alpha_3\alpha_4\alpha_5\alpha_6\alpha_7}}\theta_{\alpha_4\alpha_5\alpha_6\alpha_7} (\delta z)_{\alpha_2\alpha_3\alpha_4}\Lambda_{\alpha_0\alpha_1\alpha_2}\mathbb{C}_{\alpha_0}\\
&+\sum_{\alpha_0,\alpha_1,\alpha_2,\alpha_3,\alpha_4,\alpha_5,\alpha_6,\alpha_7,\alpha_8\in I}\int_{\mathfrak{P}^{(0)}_{\alpha_0\alpha_1\alpha_2\alpha_3\alpha_4\alpha_5\alpha_6\alpha_7\alpha_8}}\theta_{\alpha_5\alpha_6\alpha_7\alpha_8}  (\delta z)_{\alpha_3\alpha_4\alpha_5}\Lambda_{\alpha_1\alpha_2\alpha_3}\Psi_{\alpha_0\alpha_1},
\end{align*}
where the notation $\approx$ means that we have omitted the term as the form of $(2\pi\sqrt{-1})^3\mathbb{Z}$.
The three terms on the right hand side of $\approx$ can be eliminated by the variations of the actions
\begin{align*}
  I_3&=-\sum_{\alpha_0,\alpha_1,\alpha_2,\alpha_3,\alpha_4,\alpha_5,\alpha_6\in I} \int_{\mathfrak{P}^{(2)}_{\alpha_0\alpha_1\alpha_2\alpha_3\alpha_4\alpha_5\alpha_6}}\theta_{\alpha_3\alpha_4\alpha_5\alpha_6}\Upsilon_{\alpha_1\alpha_2\alpha_3}\Theta_{\alpha_0\alpha_1}d\mathbb{C}_{\alpha_0},\\
I_4&=\sum_{\alpha_0,\alpha_1,\alpha_2,\alpha_3,\alpha_4,\alpha_5,\alpha_6,\alpha_7\in I}\int_{\mathfrak{P}^{(1)}_{\alpha_0\alpha_1\alpha_2\alpha_3\alpha_4\alpha_5\alpha_6\alpha_7}}\theta_{\alpha_4\alpha_5\alpha_6\alpha_7}\Upsilon_{\alpha_2\alpha_3\alpha_4}\Lambda_{\alpha_0\alpha_1\alpha_2} \mathbb{C}_{\alpha_0}, \\
I_5&=-\sum_{\alpha_0,\alpha_1,\alpha_2,\alpha_3,\alpha_4,\alpha_5,\alpha_6,\alpha_7,\alpha_8\in I}\int_{\mathfrak{P}^{(0)}_{\alpha_0\alpha_1\alpha_2\alpha_3\alpha_4\alpha_5\alpha_6\alpha_7\alpha_8}}\theta_{\alpha_5\alpha_6\alpha_7\alpha_8}
\Upsilon_{\alpha_3\alpha_4\alpha_5}\Lambda_{\alpha_1\alpha_2\alpha_3} \Psi_{\alpha_0\alpha_1},
\end{align*}
respectively, under the transformation $\Upsilon_{\alpha\beta\gamma}\mapsto \Upsilon_{\alpha\beta\gamma}+(\delta z)_{\alpha\beta\gamma}$. Similar calculations exhibit $I_2+I_3$ and $I_4+I_5$ are invariant under the local gauge transformations.
Consequently, the gauge invariant $U(1)$-ABC action reads
\begin{align}\label{bjk}
 S_{\mathrm{ABC}}\approx&\ I_0+I_1+I_2+I_3+I_4+I_5\nonumber\\
 =& \ \sum_{\alpha_0,\alpha_1,\alpha_2,\alpha_3\in I}\int_{\mathfrak{P}^{(5)}_{\alpha_0\alpha_1\alpha_2\alpha_3}}\theta_{\alpha_0\alpha_1\alpha_2\alpha_3}\mathbb{A}_{\alpha_0}\wedge d\mathbb{B}_{\alpha_0}\wedge d\mathbb{C}_{\alpha_0}\nonumber\\
& -\sum_{\alpha_0,\alpha_1,\alpha_2\alpha_3,\alpha_4\in I}\int_{\mathfrak{P}^{(4)}_{\alpha_0\alpha_1\alpha_2\alpha_3\alpha_4}}\theta_{\alpha_1\alpha_2\alpha_3\alpha_4}\Gamma_{\alpha_0\alpha_1}d\mathbb{B}_{\alpha_0}\wedge d\mathbb{C}_{\alpha_0}\nonumber\\
& +\sum_{\alpha_0,\alpha_1,\alpha_2,\alpha_3,\alpha_4,\alpha_5\in I}\int_{\mathfrak{P}^{(3)}_{\alpha_0\alpha_1\alpha_2\alpha_3\alpha_4\alpha_5}}\theta_{\alpha_2\alpha_3\alpha_4\alpha_5}\Upsilon_{\alpha_0\alpha_1\alpha_2}\mathbb{B}_{\alpha_0}\wedge d\mathbb{C}_{\alpha_0}\nonumber\\
&-\sum_{\alpha_0,\alpha_1,\alpha_2,\alpha_3,\alpha_4,\alpha_5,\alpha_6\in I} \int_{\mathfrak{P}^{(2)}_{\alpha_0\alpha_1\alpha_2\alpha_3\alpha_4\alpha_5\alpha_6}}\theta_{\alpha_3\alpha_4\alpha_5\alpha_6}\Upsilon_{\alpha_1\alpha_2\alpha_3}\Theta_{\alpha_0\alpha_1}d\mathbb{C}_{\alpha_0}\nonumber\\
&+ \sum_{\alpha_0,\alpha_1,\alpha_2,\alpha_3,\alpha_4,\alpha_5,\alpha_6,\alpha_7\in I}\int_{\mathfrak{P}^{(1)}_{\alpha_0\alpha_1\alpha_2\alpha_3\alpha_4\alpha_5\alpha_6\alpha_7}}\theta_{\alpha_4\alpha_5\alpha_6\alpha_7}\Upsilon_{\alpha_2\alpha_3\alpha_4}\Lambda_{\alpha_0\alpha_1\alpha_2} \mathbb{C}_{\alpha_0}\nonumber\\
&-\sum_{\alpha_0,\alpha_1,\alpha_2,\alpha_3,\alpha_4,\alpha_5,\alpha_6,\alpha_7,\alpha_8\in I}\int_{\mathfrak{P}^{(0)}_{\alpha_0\alpha_1\alpha_2\alpha_3\alpha_4\alpha_5\alpha_6\alpha_7\alpha_8}}\theta_{\alpha_5\alpha_6\alpha_7\alpha_8}
\Upsilon_{\alpha_3\alpha_4\alpha_5}\Lambda_{\alpha_1\alpha_2\alpha_3} \Psi_{\alpha_0\alpha_1}.
\end{align}

The  Deligne-Beilinson cup product \cite{dbb}
\begin{align*}
 \bigcup:&\ H^q_{\textrm{DB}}(X,\mathbb{Z}(\ell))\times H^t_{\mathrm{DB}}(X,\mathbb{Z}(\jmath))\\
&\ \rightarrow \left\{
                                                      \begin{array}{ll}
                                                                                                            H^{q+t}_{\textrm{DB}}(X,\mathbb{Z}(\ell+\jmath+1))  , & \hbox{$q=\ell+1, t\leq \jmath+1$\textrm{ or } $t=\jmath+1, q\leq \ell+1$;} \\
                                                                                H^{q+t-1}_{\textrm{DB}}(X,\mathbb{Z}(\ell+\jmath+1))                             , & \hbox{$q\geq\ell+2, t\leq \jmath+1$\textrm{ or } $t\geq\jmath+2, q\leq \ell+1$;} \\
                                                                                                            H^{q+t}_{\textrm{DB}}(X,\mathbb{Z}(\ell+\jmath+1)) \simeq H^{q+t}(X,\mathbb{Z})                            , & \hbox{$q\geq\ell+2, t\geq \jmath+2$\textrm{ or } $t\geq\jmath+2, q\geq \ell+2$;} \\
0, & \hbox{\textrm{ other cases. }}
                                                                                                          \end{array}
                                                                                                        \right.
\end{align*}
defines
\begin{align*}
\mathcal{A}\bigcup \mathcal{B} \in  & \ H^{4}_{\textrm{DB}}(\mathcal{M},\mathbb{Z}(3)), \ \ \mathcal{A}\bigcup \mathcal{B}\bigcup\mathcal{C}\in H^{6}_{\textrm{DB}}(\mathcal{M},\mathbb{Z}(5)),\\
& \ \ \varphi\bigcup\mathcal{A} \bigcup \mathcal{B}\bigcup \mathcal{C}\in H^{8}_{\textrm{DB}}(\mathcal{M},\mathbb{Z}(7))\simeq \mathbb{R}/\mathbb{Z}.
\end{align*}
Then the action \eqref{bjk} can be briefly written as
\begin{align}\label{mnb}
 S_{\mathrm{ABC}}\approx&\ \frac{1}{2\pi\sqrt{-1}}\int_\mathcal{M}\varphi \bigcup\mathcal{A} \bigcup \mathcal{B}\bigcup \mathcal{C}\nonumber\\
& -\frac{1}{2\pi\sqrt{-1}}\sum_{\alpha_0,\alpha_1\in I}\int_{\mathfrak{P}^{(7)}_{\alpha_0\alpha_1}}\chi_{\alpha_0\alpha_1}d\mathbb{A}_{\alpha_0}\wedge d\mathbb{B}_{\alpha_0}\wedge d\mathbb{C}_{\alpha_0}\nonumber\\
&+\frac{1}{2\pi\sqrt{-1}}\sum_{\alpha_0,\alpha_1,\alpha_2\in I}\int_{\mathfrak{P}^{(6)}_{\alpha_0\alpha_1\alpha_2}}\tau_{\alpha_0\alpha_1\alpha_2}d\mathbb{A}_{\alpha_0}\wedge d\mathbb{B}_{\alpha_0}\wedge d\mathbb{C}_{\alpha_0},
\end{align}
where the extra two terms are obviously gauge invariant.

\end{document}